\documentclass[preprint,review,11pt]{elsarticle}

\journal{Reliability Engineering \& System Safety}
\usepackage{graphicx}

\usepackage{fullpage}
\usepackage{color,soul}
\usepackage{ulem}
\usepackage{rotating}
\usepackage{graphicx,psfrag,epsf}
\usepackage{enumerate}
\usepackage{float}
\usepackage{lineno,xcolor}

\usepackage{amsmath}
\usepackage{amssymb}

\usepackage{algorithm}
\usepackage{algorithmic}

\usepackage{subfigure}

\newcommand{\blind}{0}

\newtheorem{theorem}{Theorem}[section]

\newenvironment{proof}[1][Proof]{\begin{trivlist}
\item[\hskip \labelsep {\bfseries #1}]}{\end{trivlist}}

\begin{document}

\if0\blind
{
\title{Refined Stratified Sampling for efficient Monte Carlo based uncertainty quantification}
\author{Michael D. Shields$^{1*}$, \& Kirubel Teferra$^{1}$  \& Adam Hapij$^{2}$  \& Raymond P. Daddazio$^{2}$ \\
		$~^{1}$ Dept. of Civil Eng., Johns Hopkins University \\
		$~^{*}$ Corresponding Author: michael.shields@jhu.edu \\
		$~^{2}$ Applied Science \& Investigations, Weidlinger Associates, Inc. }

} \fi

\if1\blind
{
\title{Refined Stratified Sampling for efficient Monte Carlo based uncertainty quantification}
} \fi

\bigskip


\begin{abstract}
A general adaptive approach rooted in stratified sampling (SS) is proposed for sample-based uncertainty quantification (UQ). To motivate its use in this context the space-filling, orthogonality, and projective properties of SS are compared with simple random sampling and Latin hypercube sampling (LHS). SS is demonstrated to provide attractive properties for certain classes of problems. The proposed approach, Refined Stratified Sampling (RSS), capitalizes on these properties through an adaptive process that adds samples sequentially by dividing the existing subspaces of a stratified design. RSS is proven to reduce variance compared to traditional stratified sample extension methods while providing comparable or enhanced variance reduction when compared to sample size extension methods for LHS - which do not afford the same degree of flexibility to facilitate a truly adaptive UQ process. An initial investigation of optimal stratification is presented and motivates the potential for major advances in variance reduction through optimally designed RSS. Potential paths for extension of the method to high dimension are discussed. Two examples are provided. The first involves UQ for a low dimensional function where convergence is evaluated analytically. The second presents a study to asses the response variability of a floating structure to an underwater shock.
\end{abstract}


\maketitle

%
%

\section{Introduction}
Monte Carlo methods are used for uncertainty quantification (UQ) in nearly every field of engineering and computational science. They are often (rightfully) criticized for their high computational cost - especially for reliability analysis and other applications where extreme response is of interest. Yet, they remain the most robust and effective means of quantifying uncertainty in computational analyses. For this reason, there has been much interest in improving the computational efficiency of Monte Carlo methods. As Janssen \cite{Janssen_RESS_13} points out, there are two means of accomplishing this computational savings: 1.\ Improve the convergence rate of the sampling routine; and/or 2.\ Perform a sample set that is optimally small. This work addresses both of these aspects.

To achieve an improved convergence rate, we focus on a class of variance reduction techniques called stratified sampling \citep{Tocher_63,Helton_Davis_RESS_03} that includes the popular Latin hypercube sampling (LHS) method \citep{McKay_et_al_Tech_79,Helton_Davis_RESS_03}. Stratified sampling methods operate by subdividing the sample space into smaller regions and sampling within these regions. In so doing, the produced samples more effectively fill the sample space and therefore reduce the variance of computed statistical estimators. LHS in particular has been widely used for uncertainty quantification (e.g. \citep{Olsson_et_al_SS_03,Wang_JMD_03}) and its properties have been studied extensively \citep{McKay_et_al_Tech_79,Stein_Tech_87,Owen_JRSSB_92,Huntington_Lyrintzis_PEM_98,Fang_Ma_JOC_01,Fang_et_al_MoC_02}. True stratified sampling, on the other hand, has not received a similar level of attention in the UQ literature - largely due to the desirable properties of LHS for many applications - although it is widely recognized as an effective means of variance reduction \cite{Glasserman_2003,Rubinstein_Kroese_08, Kalos_Whitlock_2008}. In this work, some properties of stratified sampling methods (including LHS) are studied and the case is made for the use of true stratified sampling for certain classes of problems. In particular, there are many low-to-moderate dimensional applications where true stratified sampling is competitive with or even more effective than LHS. This motivates the need to extend these benefits to high dimensional problems and some insights for future work along these lines are provided.

As for minimizing the sample size, there is no generally accepted means to identify an optimally small sample set {\it{a priori}}. In this work, we espouse a general adaptive paradigm for UQ (Figure \ref{fig:UQ_Paradigm}) wherein samples are progressively added and analyses conducted until a user-defined convergence condition is met. This is not a new concept, but it is one that is perhaps underdeveloped. Early work by Gilman \citep{Gilman_Conf_68} demonstrated this idea for classical Monte Carlo analyses using simple random sampling (SRS) and a recent work by Janssen \citep{Janssen_RESS_13} provides a renewed emphasis on its importance. Until recently though, most stratified sampling methods required {\it{a priori}} prescription of the sample size and precluded sample size extension. Therefore, it was necessary to conservatively oversample - otherwise the entire analysis would need to be restarted if the convergence criteria were not met. 
\begin{figure}[H]
\centering
\includegraphics[width=1.\columnwidth]{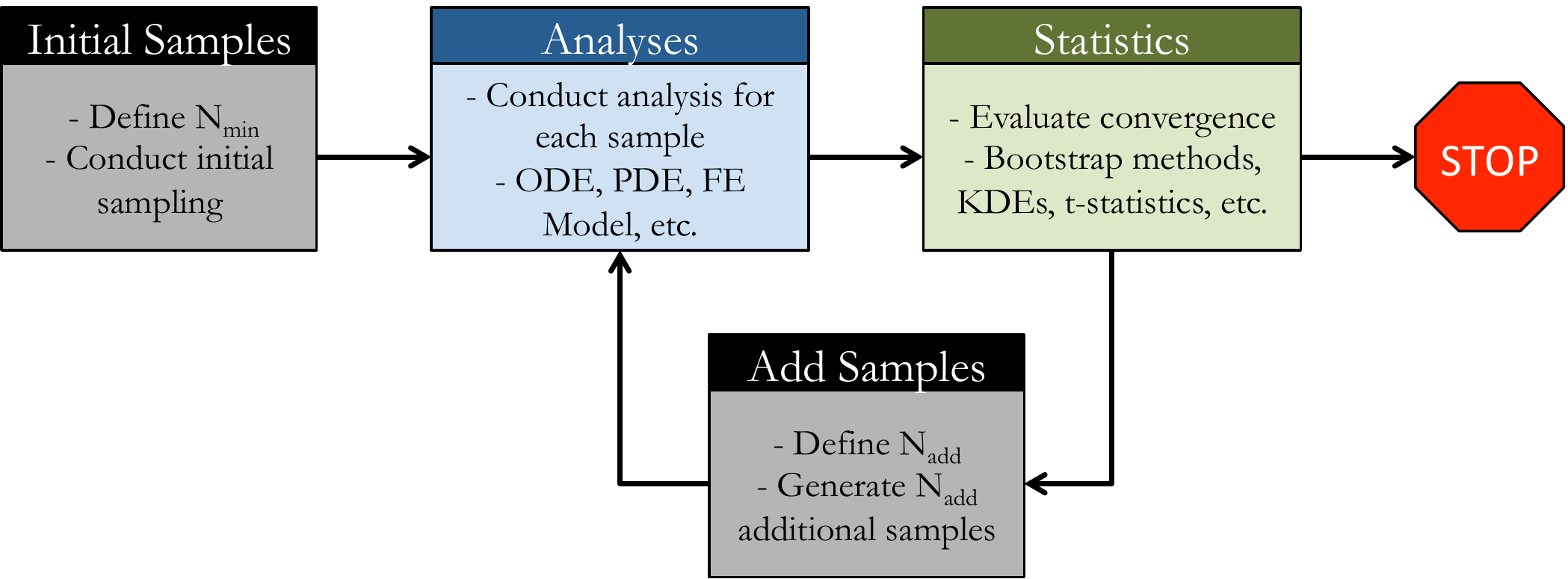}
\caption{General adaptive process for UQ.}
\label{fig:UQ_Paradigm}
\end{figure}

Recent developments in LHS enable sample size extension using so-called Hierarchical Latin Hypercube Sampling (HLHS) \citep{Tong_RESS_06,Sallaberry_et_al_RESS_08,Vorechovsky_et_al_ICOSSAR_13}, Replicated Latin Hypercubes (RLHs) \citep{Iman_Conf_81,McKay_TR_95}, and related methods \citep{Qian_Bio_09,Xiong_EO_09,Rennen_et_al_SMO_10}. These methods, while important, do not afford the flexibility necessary to fully capitalize on the UQ paradigm promoted in this work. Hierarchical methods, for example, produce - at a minimum - a new sample for each existing sample during the extension. Consequently, sample size increases exponentially with the number of extensions performed. RLHs grow linearly with the number of extensions ($n_{i+1}=2n_{i}$) and require each subsequent LHS to be the same size and utilize the same stratification as the original. Since there is no refinement of the strata using RLHs, there is a tradeoff between loss of desirable sample properties and sample size growth rate. These considerations limit the adoption of a truly adaptive UQ process.

A few adaptive SS methods have been proposed in the past - primarily in the physics community focusing on applications in multidimensional integration \cite{Lepage_JCP_78,Press_Farrar_CiP_90}. However, these methods have slightly different aims and, again, do not possess the desired level of adaptability. To enable effective sample size extension from a stratified design, a new methodology - Refined Stratified Sampling (RSS) - is proposed that is conceptually simple and adds as many or as few samples as desired at each extension. The method operates by dividing existing strata and generating samples in the newly created empty strata. The method is proven to unconditionally reduce the variance of the sample set when compared with existing methods for sample size extension in stratified sampling. It is demonstrated that the method affords convergence rates for statistical estimators that are comparable or superior to LHS for certain classes of problems with the added advantage of maximal flexibility in its sample size extension capabilities. Moreover, through an exploration of stratum optimality, we suggest that the RSS method provides a convenient and rigorous avenue for identifying a stratification of the input space that samples optimally in the output space of interest. This is the basis of a parallel effort for applications in reliability analysis \cite{Shields_Sundar_IJRS_2015}.

The emphasis of this paper is on developing the RSS methodology, which is then demonstrated on low-dimensional applications where the benefits of true stratified sampling are observed. Considerations for its extension to high dimensional problems are discussed and motivation provided for future research along these lines. Two specific example applications are provided.

\section{Review of Sampling Methods}

Consider a stochastic system described by the relation:
\begin{linenomath}
\begin{equation}
\label{eqn:general_system}
Y=F(\mathbf{X})
\end{equation}
\end{linenomath}
where the random vector $\mathbf{X}=\{X_1,X_2,\dots,X_n\}$ possesses independent components defined over the sample space $\mathbf{\mathcal{S}}$ describing $n$ input random variables with marginal cumulative distribution functions (CDFs) $D_{X_i}(\cdot)$. Correlated random variables are not considered in this work as it is common practice to produce a set of uncorrelated random variables from a correlated set using methods such as Principal Component Analysis and the Nataf or Rosenblatt transformations. The operator $F(\cdot)$ commonly represents a computer simulation such as a finite element model possessing strong nonlinearities and/or instabilities such that $Y$ is difficult to assess probabilistically. The following presents a brief review of common Monte Carlo methods for randomly sampling $\mathbf{X}$ in order to perform statistical analysis of $Y$.

For notational clarity, subscript $i$ is used to denote a vector component, subscript $k$ is used to denote a stratum of the space, and subscript $l$ denotes a specific sample. Additionally, $n$ refers to the total number of vector components (dimension), $M$ refers to the number of strata in a design, and $N$ refers to the total number of samples.


\subsection{Simple Random Sampling}

Traditional Monte Carlo methods rely on so-called Simple Random Sampling (SRS) or Monte Carlo Sampling in which realizations of $\mathbf{X}$, denoted $\mathbf{x}_l;\hspace{3pt} l=1,\cdots,N$ (samples), are generated as independent and identically distributed (iid) realizations on $\mathbf{\mathcal{S}}$ by:
\begin{linenomath}
\begin{equation}
\label{eqn:1}
x_{li}=D_{X_{i}}^{-1}(U_i); i=1,2,\dots,n
\end{equation} 
\end{linenomath}
where $U_i$ are iid uniformly distributed samples on $[0,1]$. The realizations $\mathbf{x}$ are then applied to the system $y=F(\mathbf{x})$ and $y$ is statistically evaluated. 


\subsection{Stratified Sampling}
\label{sec:SS}
Stratified Sampling (SS) divides the sample space $\mathbf{\mathcal{S}}$ into a collection of $M$ disjoint subsets (strata) $\mathbf{\Omega}_k;  k=1,2,\dots,M$ with $\cup_{k=1}^M\mathbf{\Omega}_k=\mathbf{\mathcal{S}}$ and $\mathbf{\Omega}_p\cap\mathbf{\Omega}_q=\emptyset; p\ne q$. Samples $\mathbf{x}_l=[x_{l1},x_{l2},\dots,x_{ln}]; l=1,2,\dots,N$ are generated by randomly drawing $M_k$ ($\sum_{k=1}^M M_k=N$) samples within each stratum $k$ according to:
\begin{linenomath}
\begin{equation}
\label{eqn:2}
x_{li}^{(k)}=D_{X_{i}}^{-1}(U_{ik}); i=1,2,\dots,n
\end{equation} 
\end{linenomath}
where $U_{ik}$ are iid uniformly distributed samples on $[\xi_{ik}^{lo},\xi_{ik}^{hi}]$ with $\xi_{ik}^{lo}=D_{X_{i}}(\zeta_{ik}^{lo})$ and $\xi_{ik}^{hi}=D_{X_{i}}(\zeta_{ik}^{hi})$ and $\zeta_{ik}^{lo}$ and $\zeta_{ik}^{hi}$ denote the lower and upper bounds respectively of the $i^{th}$ vector component of stratum $\mathbf{\Omega}_k$. For our purposes, stratification is performed directly in the probability space meaning that the strata are defined by prescribing the bounds $\xi_{ik}^{lo}$ and $\xi_{ik}^{hi}$ on the n-dimensional unit hypercube. 

Beyond being disjoint ($\mathbf{\Omega}_p\cap\mathbf{\Omega}_q=\emptyset; p\ne q$) and filling the space ($\cup_{k=1}^M\mathbf{\Omega_k}=\mathbf{\mathcal{S}}$), there are no restrictions on the strata definitions. Strata can be defined with different sizes and shapes in general and can even be defined a posteriori based on an existing sample set (so-called post stratification) \citep{Tocher_63}. The size of the strata in the probability space, $p_k$, is equal to its probability of occurrence $p_k=P[\mathbf{\Omega}_k]$. With these considerations in mind, we define three general classes of stratified designs.
\begin{enumerate}
\item \emph{Symmetrically Balanced Stratified Design} (SBSD): A stratified design is said to be symmetrically balanced if all strata possess equal probability $p_i=p_j;\hspace{3pt} \forall i,j$ and the probability space is divided equally in all variables of the design. In an SBSD, each stratum is an n-dimensional hypercube of equal size.
\item \emph{Asymmetrically Balanced Stratified Design} (ABSD): A stratified design is said to be asymmetrically balanced if all strata possess equal probability $p_i=p_j;\hspace{3pt} \forall i,j$ but the strata limits are assigned arbitrarily. In the simplest case of a ABSD, each stratum is an n-dimensional orthotope (or hyperrectangle). However, general ABSDs may possess polyhedral or other arbitrarily shaped strata.
\item \emph{Unbalanced Stratified Design} (UBSD): A stratified design is said to be unbalanced if all strata do not possess equal probability (i.e. $\exists\hspace{3pt} p_i,p_j:p_i\ne p_j$).
\end{enumerate}
We refer to these terms throughout the paper and samples drawn from these designs are referred to as Symmetrically Balanced Stratified Samples (SBSS), Asymmetrically Balanced Stratified Samples (ABSS), and Unbalanced Stratified Samples (UBSS) respectively.


\subsection{Latin Hypercube Sampling}
Latin Hypercube Sampling (LHS) divides the range of each vector component $\mathbf{X}_i; i=1,2,\dots,n$ into $M$ disjoint subsets (strata) of \emph{equal probability} $\mathbf{\Omega}_{ik};  i=1,2,\dots,n; k=1,2,\dots,M$. Samples of each vector component are drawn from the respective strata according to:
\begin{linenomath}
\begin{equation}
\label{eqn:5}
x_{li}^{(k)}=D_{X_{i}}^{-1}(U_{ik}); i=1,2,\dots,n;  k=1,2,\dots,M
\end{equation} 
\end{linenomath}
where $U_{ik}$ are iid uniformly distributed samples on $[\xi_{k}^{lo},\xi_{k}^{hi}]$ with $\xi_{k}^{lo}=\dfrac{k-1}{M}$ and $\xi_{k}^{hi}=\dfrac{k}{M}$. The samples $\mathbf{x}_l=[x_{l1},x_{l2},\dots,x_{ln}]; l=1,2,\dots,N$ are assembled by (uniformly) randomly grouping the terms of the generated vector components. That is, a term $x_{li}$ is generated by randomly selecting from the generated components $x_{li}^{(k)}$ (without replacement) and these terms are grouped to produce a sample. This process is repeated $M$ times. 

Because the component samples are randomly paired, an LHS is not unique; there are $(M!)^{n-1}$ possible combinations. Improved LHS algorithms have been developed to determine optimal pairings that either enhance space-filling or reduce spurious correlation (increasing orthogonality). To improve space-filling, several Latin hypercube design methods have been developed that minimize the $L_2$-discrepancy \citep{Fang_et_al_MoC_02}, maximize the minimum distance between points (`maximin' designs) \citep{Johnson_et_al_JSPI_90,Morris_Mitchell_JSPI_95,Ye_et_al_JSPI_2000,Joseph_Hung_SS_08}, optimize projection properties to ensure samples are evenly spread when projected onto a known subspace \citep{Liefvendahl_Stocki_JSPI_06}, and minimize integrated mean square error and maximize entropy \citep{Park_JSPI_94}, among others. Similarly, numerous efforts have been made to reduce spurious correlations beginning with Iman and Connover \citep{Iman_Conover_CSSC_82} and later Florian \citep{Florian_PEM_92} who rearrange the matrix of samples based on a transformation of the rank number matrix. Huntington and Lyrintzis \citep{Huntington_Lyrintzis_PEM_98} and Vorechovsky and Novak \citep{Vorechovsky_Novak_PEM_09} utilized iterative optimization methods to reduce spurious correlations while others use orthogonal arrays \citep{Tang_JASA_93}. Further, several authors have developed methods for constructing orthogonal Latin hypercubes that additionally possess enhanced space-filling properties \citep{Ye_JASA_98,Cioppa_Lucas_Tech_07}. Many of these designs are complex, laborious to implement, and/or computationally intensive. Meanwhile, there is no consensus on whether space-filling or minimal spurious correlations are preferable but certainly the two properties are linked and there may be a trade-off in achieving either. 



\subsection{Statistical evaluation \& variance reduction}
\label{sec:stat_eval}
For the purposes of evaluating statistical properties of the different sampling methods, consider the general statistical estimator defined by:
\begin{linenomath}
\begin{equation}
T(y_1, \dots, y_N) = \sum\limits_{l=1}^Nw_lg(y_l)
\end{equation}
\end{linenomath}
where $y_l=F(\mathbf{x}_l)$ and $\mathbf{x}_l$ denotes a sample generated according to SRS, LHS, or SS, $w_l$ are weights attributed to each sample ($w_l=\frac{1}{N}$ for SRS and LHS but may differ for SS - see below), and $g(\cdot)$ is an arbitrary function. If $g(y)=y^r$, then $T$ represents an estimate of the $r^{th}$ moment while $g(y)=\mathbf{1}\{y\le Y\}$, where $\mathbf{1}\{\cdot\}$ denotes the indicator function, specifies estimation of the empirical CDF. In keeping with past notational conventions, we denote $T_R$, $T_L$, and $T_S$ as the statistical estimates produced from SRS, LHS, and SS respectively. 

For SRS, it is well known that the sample variance of the statistical estimator is given by:
\begin{linenomath}
\begin{equation}
\text{Var}\left[T_R\right]=\dfrac{\sigma^2}{N}
\end{equation}
\end{linenomath}
where $\sigma^2$ denotes the variance of $g(Y)$. This estimate serves as a benchmark for comparison with the considered variance reduction techniques.

Stratifed sampling does not necessitate an equal probability of occurrence for each sample. Consequently, in order to statistically evaluate the response quantity $Y$ generated using stratified samples, probabilistic weights are assigned to each sample according to the probability of occurrence of the stratum being sampled. That is, for a sample $\mathbf{x}_l\in\mathbf{\Omega}_k$, the sample weight is defined as:
\begin{linenomath}
\begin{equation}
w_l=\dfrac{P[\mathbf{\Omega}_k]}{M_k}=\dfrac{p_k}{M_k}
\label{eqn:3}
\end{equation}
\end{linenomath}
where $M_k$ is the total number of samples in $\mathbf{\Omega}_k$ subject to $\sum_{k=1}^M M_k=N$ and $\sum_{k=1}^M{P[\mathbf{\Omega}_k]}=1$. 
Stratified sampling has been proven to unconditionally reduce the variance of statistical estimators when compared to SRS. 
McKay et al.\citep{McKay_et_al_Tech_79} have shown that, for a balanced stratified design (SBSD or ABSD):
\begin{linenomath}
\begin{equation}
\text{Var}[T_S]=\text{Var}[T_R]-\dfrac{1}{N}\sum_{k=1}^Mp_k(\mu_k-\tau)^2
\label{eqn:8}
\end{equation}
\end{linenomath}
where $\mu_k$ is the mean value of the response evaluated over stratum $k$, and $\tau$ is the overall response mean. 

Similarly, McKay et al.\ \citep{McKay_et_al_Tech_79} have shown that the variance of the statistical estimator produced from a Latin hypercube design ($T_L$) is given by :
\begin{linenomath}
\begin{equation}
\text{Var}[T_L]=\text{Var}[T_R]+\dfrac{N-1}{N}\dfrac{1}{N^n(N-1)^n}\sum_R(\mu_p-\tau)(\mu_q-\tau)
\label{eqn:9}
\end{equation}
\end{linenomath}
where $\tau$ is defined as before, $\mu_p$ is the mean value of LHS cell $p$ defined by the bounds of the strata on each of the $n$ marginal distributions for sample $l$, and $R$ denotes the summation over the restricted space of $N^n(N-1)^n$ pairs $(\mu_p,\mu_q)$ of cells having no cell coordinates in common. Notice that LHS does not unconditionally reduce the variance of the estimate although it has been shown that the variance is reduced for any function $F(\cdot)$ having finite second moment when $N>>n$ \citep{Stein_Tech_87}. Stein \citep{Stein_Tech_87} has also shown that the closer $F(\mathbf{X})$ is to additive ($F(\mathbf{X})=\sum\limits_{i=1}^nF_i(X_i)$), the more the variance of $T_L$ is reduced. 

To the authors' knowledge, no comprehensive study exists which compares the variance reductions from Eqns.\ \eqref{eqn:8} and \eqref{eqn:9}. Although such a study is beyond the scope of this paper, it is immediately clear that SS and LHS reduce variance through different statistical mechanisms. Thus, different classes of problems exist for which each method will afford superior variance reduction. LHS, for example, has been demonstrated to perform well for applications where $F(\mathbf{X})$ has strong additive components while SS will be shown to perform well for applications with strong interactions among the variables of $\mathbf{X}$.

%
%

\section{Space-filling, orthogonality, and projective properties of sample designs}
\label{Sec:comp}

In this section, we study and compare the space-filling, orthogonality, and projection properties of the designs considered in the previous section. It is demonstrated that, under certain conditions, SS has merits that, to date, have not been exploited and can lead to meaningful improvements in sample design.



\subsection{Space-Filling}
\label{sec:space-filling}

Numerous metrics have been developed to quantify the space-filling of a sample design including various discrepancy measures and the maximin and minimax distances. Each metric provides a slightly different interpretation of space-filling and can therefore yield apparently inconsistent results (one design may be ``good" by one measure and ``poor" by another). It is not our intention to discuss the relative merits of these measures in general but the interested reader is referred to the works of Fang et al.\ \cite{Fang_et_al_06} and Dalbey and Karystinos \cite{Dalbey_Karystinos_AIAA_10} for further discussion. Instead, we aim to compare SRS, LHS, and SS using a couple of metrics to gain some insight into how these methods fill the space. For this purpose, we will utilize the commonly employed wrap-around $L_2$-discrepancy \cite{Hickernell_98} and a new method derived from the Voronoi decomposition \citep{Aurenhammer_ACS_91}.

The Voronoi decomposition is defined such that each point in space is associated with the nearest sample point. That is, each Voronoi cell is defined as the region of the space $\mathbf{\mathcal{P}}$, denoted $R_k$ and associated with point $P_k$, containing all points in $\mathbf{\mathcal{P}}$ whose distance to $P_k$ is less than or equal to the distance, $d$, to any other point $P_j, j\ne k$. That is:
\begin{linenomath}
\begin{equation}
R_k=\left\{\mathbf{x}\in\mathbf{X}|d(\mathbf{x},P_k)\le d(\mathbf{x},P_k)\forall j\ne k\right\}
\label{eqn:Voronoi_Cell}
\end{equation}
\end{linenomath}
where, for the purposes of this work $d(\mathbf{x},\mathbf{y})$ is the Euclidean distance. Dividing the space in this way facilitates evaluation of different sample designs because the cell sizes provide a direct measure of space-filling. In particular, large Voronoi cells correspond to regions of the space that are sparsely sampled while small cells correspond to regions of the space that are densely sampled. By this measure, an optimal space-filling design produces cells whose sizes have minimal ensemble variance. In low dimension, the Voronoi decomposition can be established explicitly. In high dimension, however, the Voronoi decomposition must be estimated implicitly. To do so, the space is populated with a large number of points ($N_p\sim10^5-10^8$) and the samples are each attributed to their nearest point in the sample set. The size (volume) of each Voronoi cell is therefore determined by the proportion of points attributed to each sample, $N_i$:
\begin{linenomath}
\begin{equation}
V_i=\dfrac{N_i}{N_p}.
\end{equation}
\end{linenomath}

With this general procedure for estimating cell sizes, we propose that the coefficient of variation of the cell sizes ($\dfrac{\sigma_{V}}{\mu_V}$), referred to as the ``V-metric" herein, provides an effective metric of space-filling that can be applied regardless of sample size or dimension. It is straightforward to interpret (low value means that the points occupy approximately the same amount of space) and is visually clear when observed in 2D. Consider 100 samples drawn using SRS, LHS, and SBSS. A plot of the Voronoi decomposition (Figure \ref{fig:voronoi_comp}) shows that the SBSS produces cells that are more uniform in size than LHS and SRS. The inset statistics verify that, indeed, SBSS produces cells whose size variance is significantly smaller than both LHS and SRS. These relations are valid regardless of sample size as demonstrated in Figure \ref{fig:voronoi_convergence}, which shows the V-metric along with maximum and minimum observed cell sizes (from all sample sets) relative to sample size as evaluated by Monte Carlo simulation averaging over 10 independent trials. Note also that, by this metric, the space-filling of LHS degrades while the SBSS and ABSS maintain their space-filling properties as the sample size grows. 

\begin{figure}
\includegraphics[width=1.0\columnwidth]{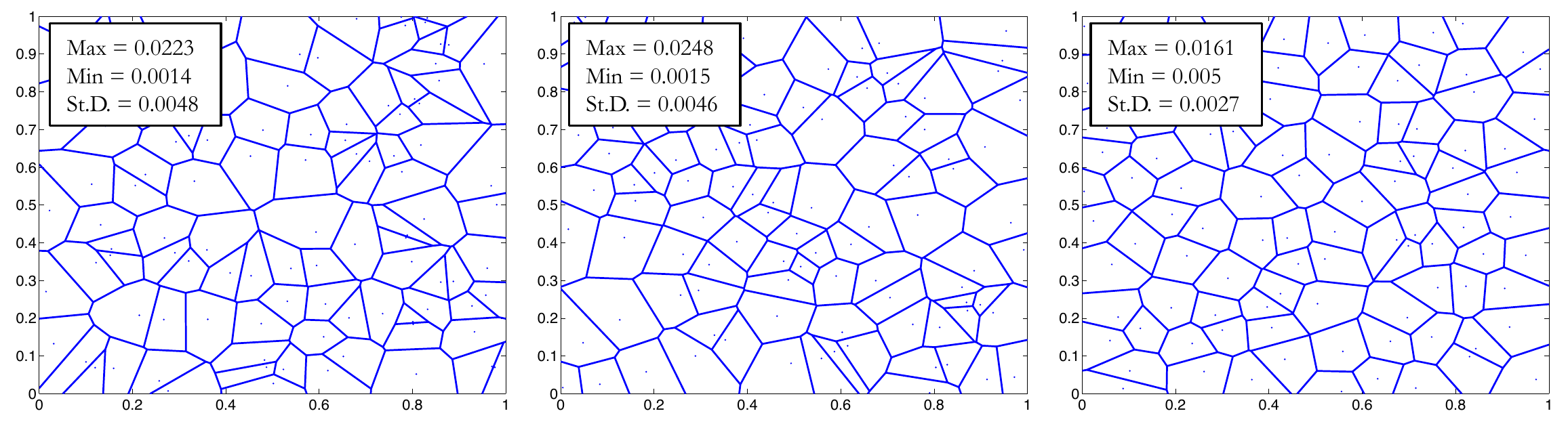}
\caption{Voronoi tessellation of the 2D probability space from 100 samples drawn using SRS (left), LHS (middle), and SBSS (right).}
\label{fig:voronoi_comp}
\end{figure}
\begin{figure}
\subfigure[\label{fig:3a}]{
\centering
\includegraphics[width=0.5\columnwidth]{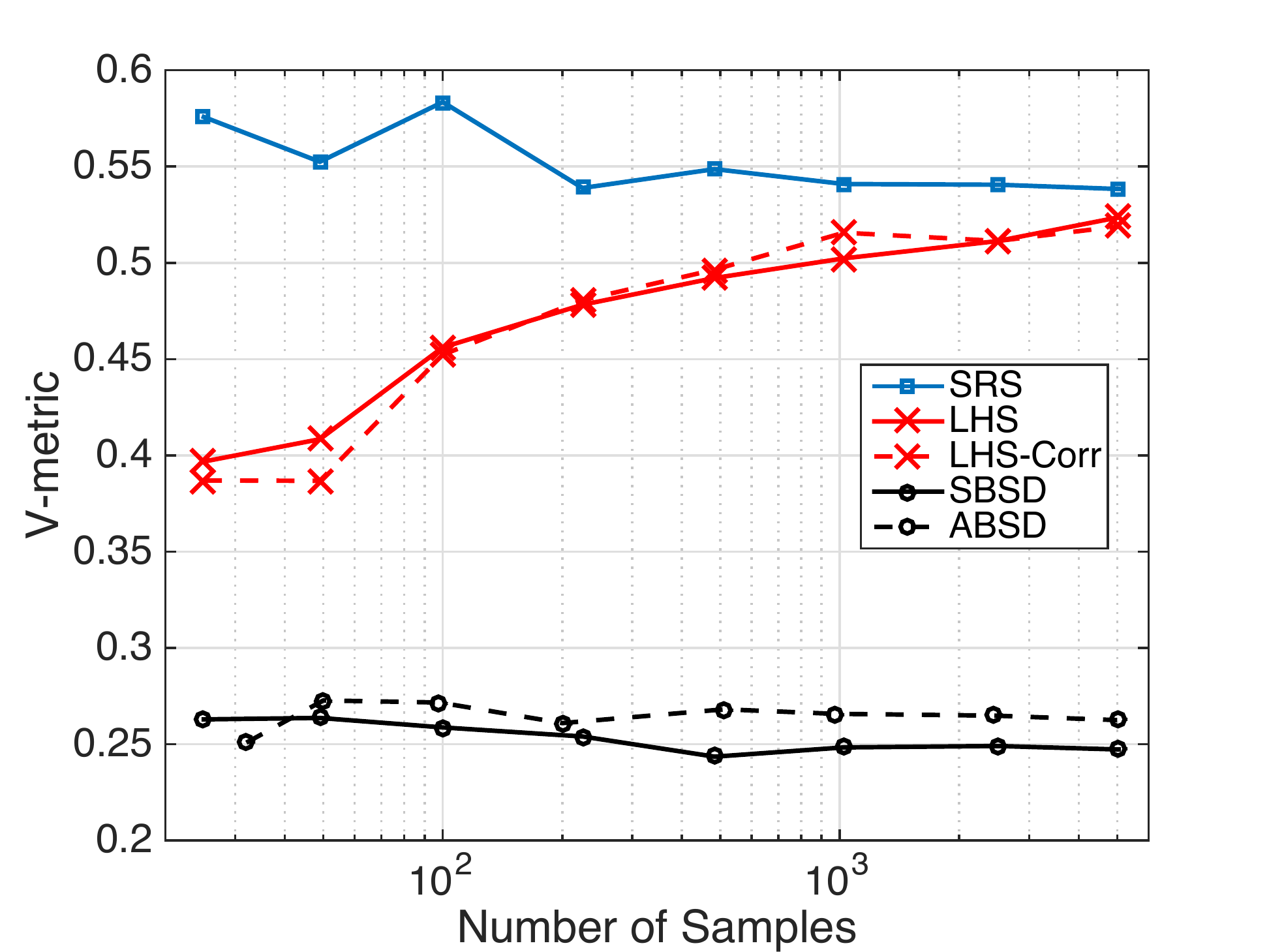}
}
\subfigure[\label{fig:3b}]{
\centering
\includegraphics[width=0.5\columnwidth]{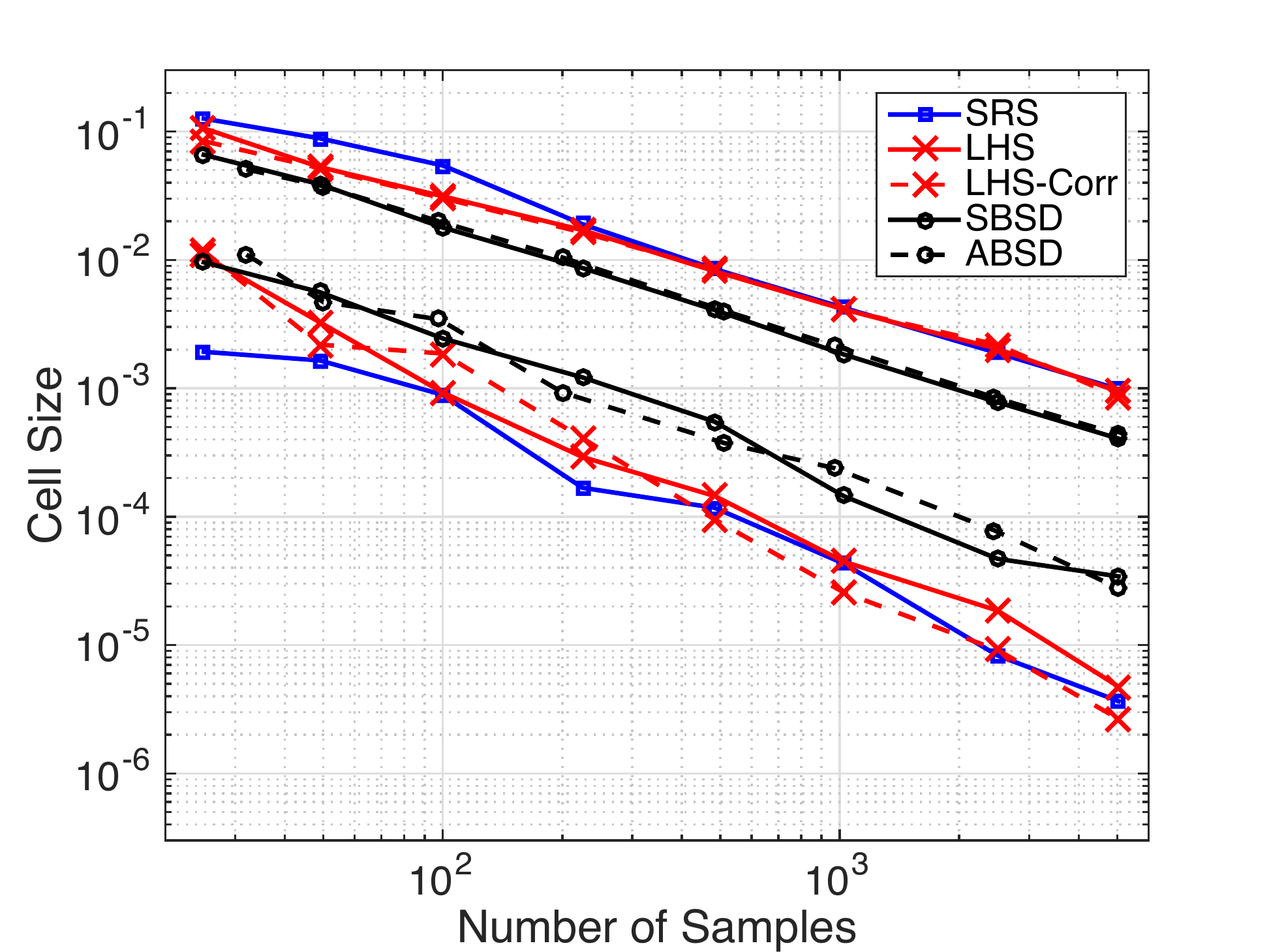}
}
\caption{(a) Voronoi cell size coefficient of variation (V-metric) and (b) maximum/minimum cell size for different sample designs as a function of sample size for 2D samples.}
\label{fig:voronoi_convergence}
\end{figure}

To study the dependence on dimension of the space-filling properties for these sample designs, we compute the average V-metric for each design from 1,024 samples (repeated 10 times) of different dimension. Figure \ref{fig:5a} shows that, by this metric, SS provides a more even distribution of samples than either LHS or SRS for low-dimensional random vectors while all sampling methods lose effectiveness for higher dimension. Notice that, for LHS, the use of the procedure by Iman and Conover \citep{Iman_Conover_CSSC_82} (denoted LHS-corr) to reduce spurious correlation has only a small effect on space-filling although, as we will see in the following section, it significantly improves the orthogonality of the sample set.

The second metric we consider is the wrap-around $L_2$-discrepancy ($D_{L_2}$) \cite{Hickernell_98} defined through the standard discrepancy measure:
\begin{linenomath}
\begin{equation}
D(\mathbf{X})=\left\| \dfrac{\mathbf{X}\cap c^M}{N}-\text{Vol}(c^M)\right\|
\end{equation}
\end{linenomath}
with $c^M\subset[0,1]^M$ such that $\|\cdot\|$ denotes an $L_2$ norm and the subset $c^M$ includes all hyperrectangles that can exist in a periodic domain $[0,1]^M$. Figure \ref{fig:5b}, which plots $D_{L_2}$ for different sample designs and dimensions, is in contrast with the Voronoi metric in that LHS is shown to provide better space-filling properties. This is due, as highlighted by Dalbey and Karystinos \cite{Dalbey_Karystinos_AIAA_10}, to the fact that $D_{L_2}$ is sensitive to differences in low-dimensional subspaces. Since LHS discretizes the 1D subspaces very evenly, $D_{L_2}$ indicates a high degree of space-filling. Meanwhile the V-metric quantifies the uniformity over the whole $M$-dimensional space.

\begin{figure}
\subfigure[\label{fig:5a}]{
\centering
\includegraphics[width=0.5\columnwidth]{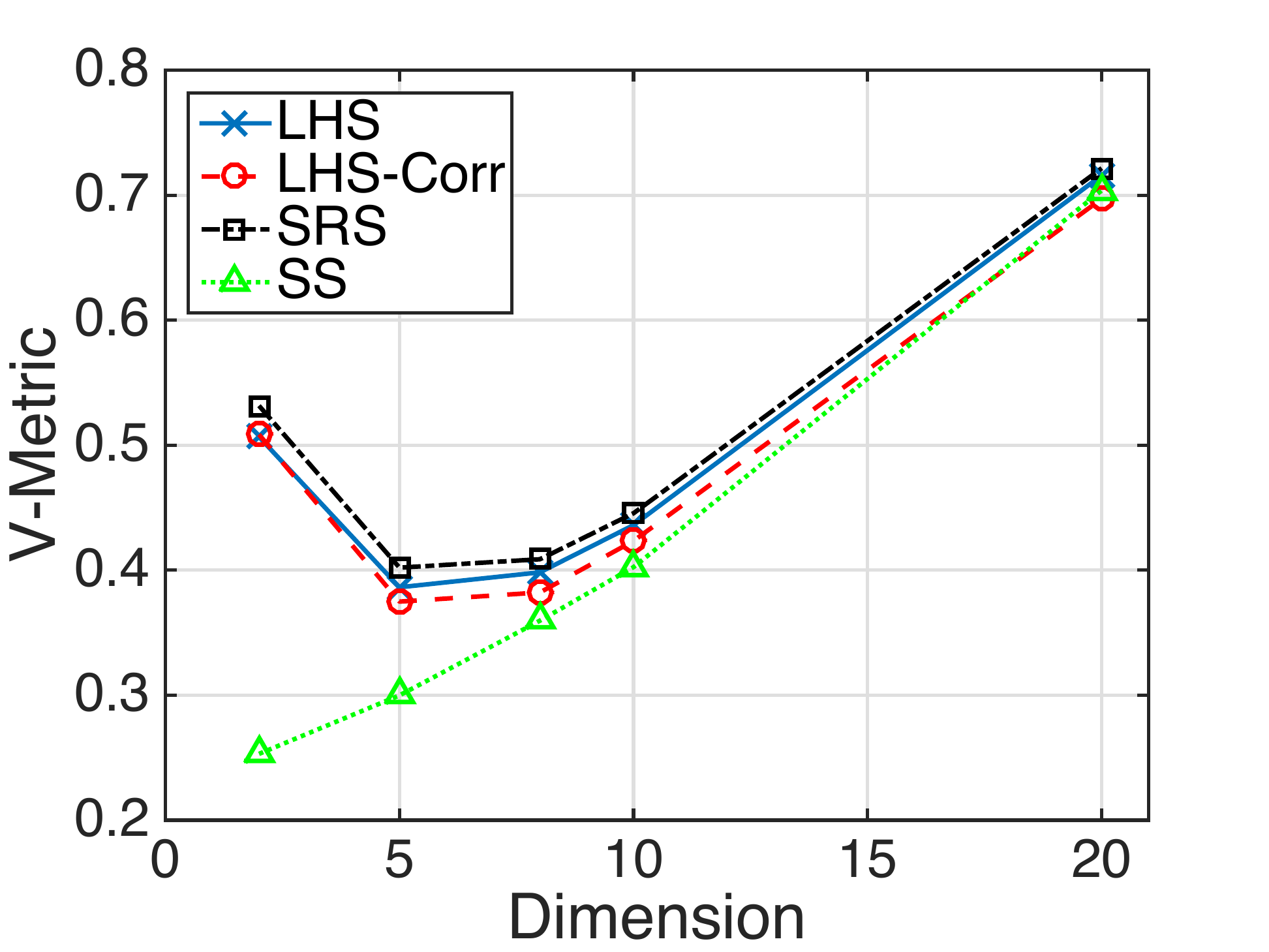}
}
\subfigure[\label{fig:5b}]{
\centering
\includegraphics[width=0.5\columnwidth]{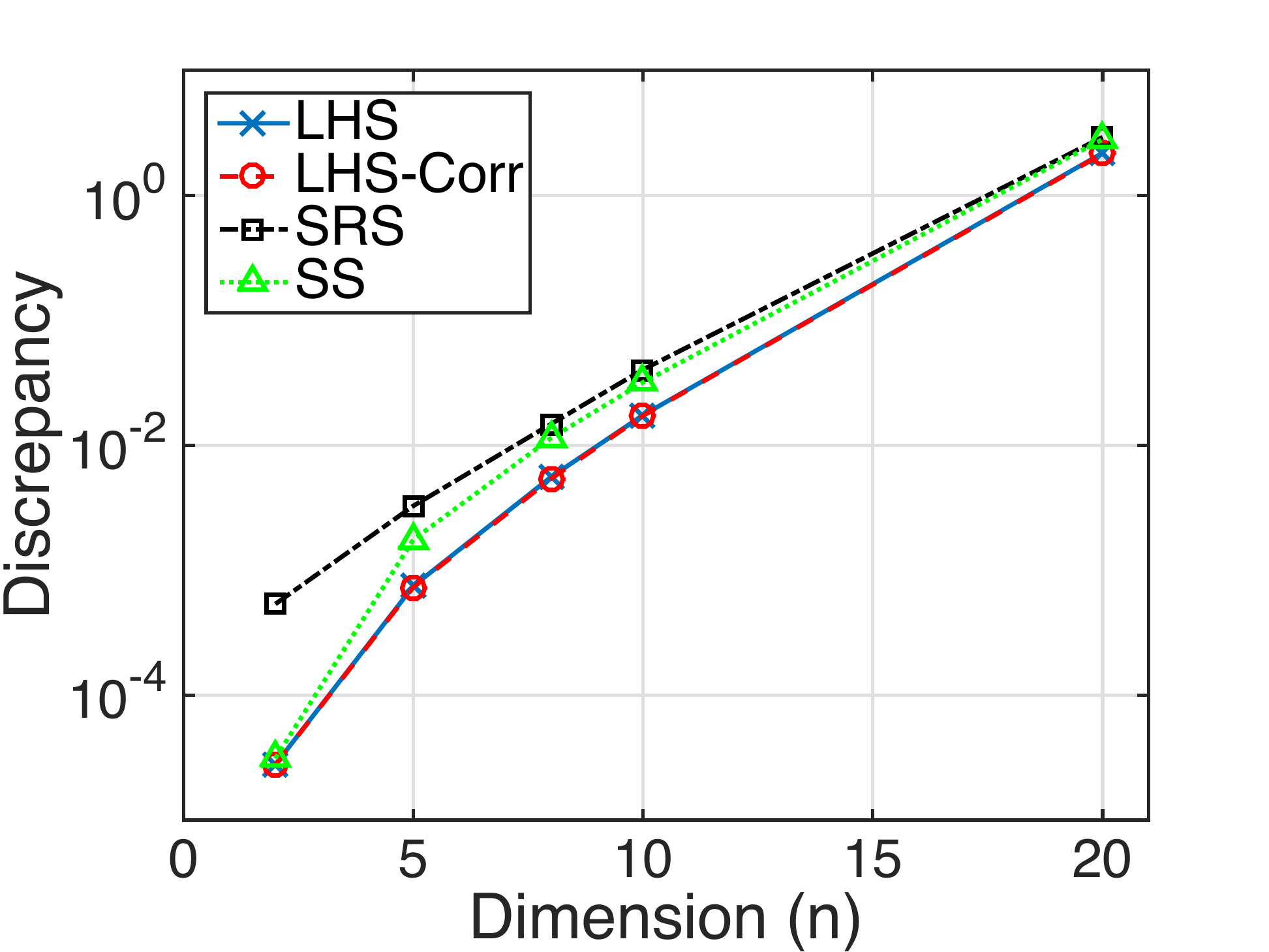}
}
\caption{Space-filling properties of LHS, LHS with correlation correction (LHS-Corr), SRS, and SS as measured using (a) the Voronoi cell metric (`V-metric') and (b) the wrap-around $L_2$ discrepancy. Note that for $n=2,5,10$ SS designs are SBSD while for $n=8,20$ SS designs are ABSD.}
\label{fig:space_filling_v_2}
\end{figure}

The comparison above does not conclude therefore that one sampling method is ``better" than the other at space-filling. Instead, it serves to highlight that LHS and SS fill the space in different ways. Owing to the hierarchical ordering principle \cite{Wu_Hamada_00}, which states that main effects and low-order interactions are usually more important than higher-order effects, the first-order space-filling of LHS is desirable for many problems. However, when interaction effects play an important role, the higher-order space-filling of SS may be desirable.

\subsection{Orthogonality}
An orthogonal sample is one whose sample correlation matrix is diagonal. That is, the random samples are perfectly uncorrelated. In practice, an orthogonal (or near orthogonal) design can be difficult to achieve. Latin hypercube designs often produce spurious correlations unless corrected or iterated in some way (e.g. \citep{Iman_Conover_CSSC_82, Florian_PEM_92, Huntington_Lyrintzis_PEM_98, Vorechovsky_Novak_PEM_09}). Stratified samples however, intuitively possess little artificial correlation to begin with. This is demonstrated in the following where we consider two measures of orthogonality similar to those considered by Cioppa and Lucas \citep{Cioppa_Lucas_Tech_07}. 

First, we evaluate the spurious correlations among the sample variates for a two dimensional sample produced using the various designs and consider their variations from the intended zero correlation 
using Monte Carlo simulation and averaging over 10 independent trials. Figure \ref{fig:6a} and \ref{fig:6a} show the standard deviation and maximum spurious correlation respectively for five different sample designs as a function of sample size. 
\begin{figure}[!ht]
\centering
\subfigure[\label{fig:6a}]{
\centering
\includegraphics[width=0.3\columnwidth]{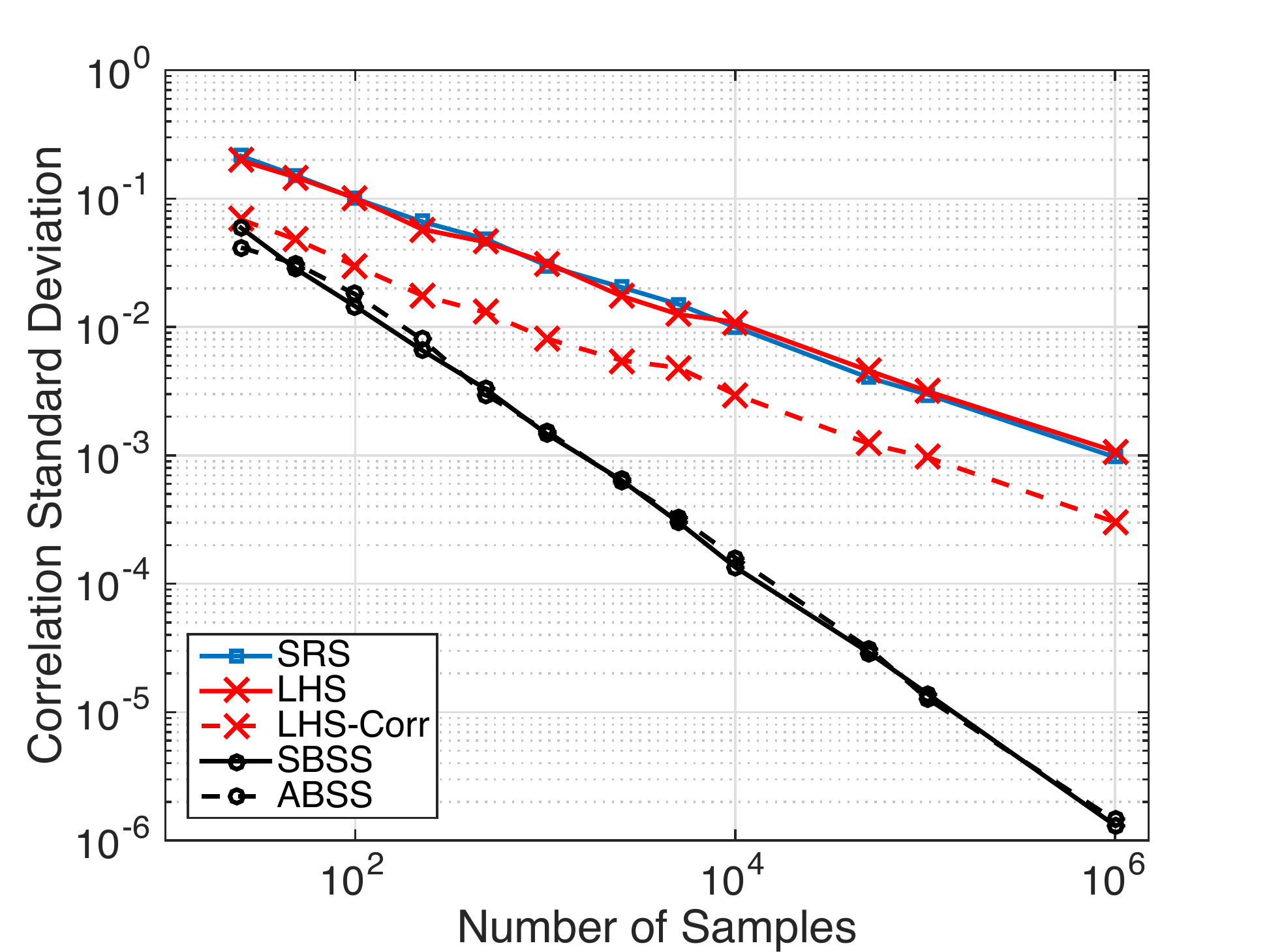}
}
\subfigure[\label{fig:6b}]{
\centering
\includegraphics[width=0.3\columnwidth]{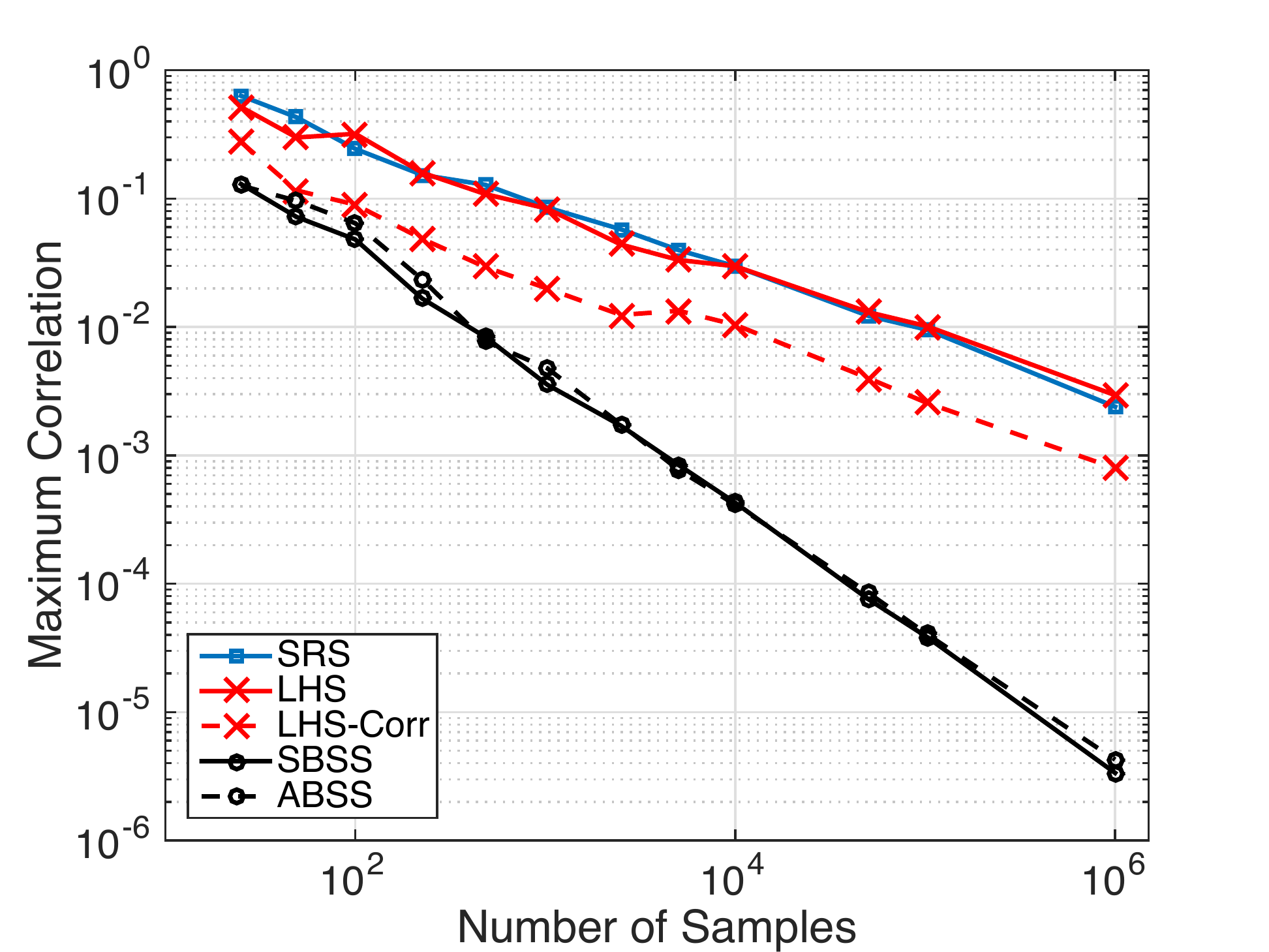}
}
\subfigure[\label{fig:6c}]{
\centering
\includegraphics[width=0.3\columnwidth]{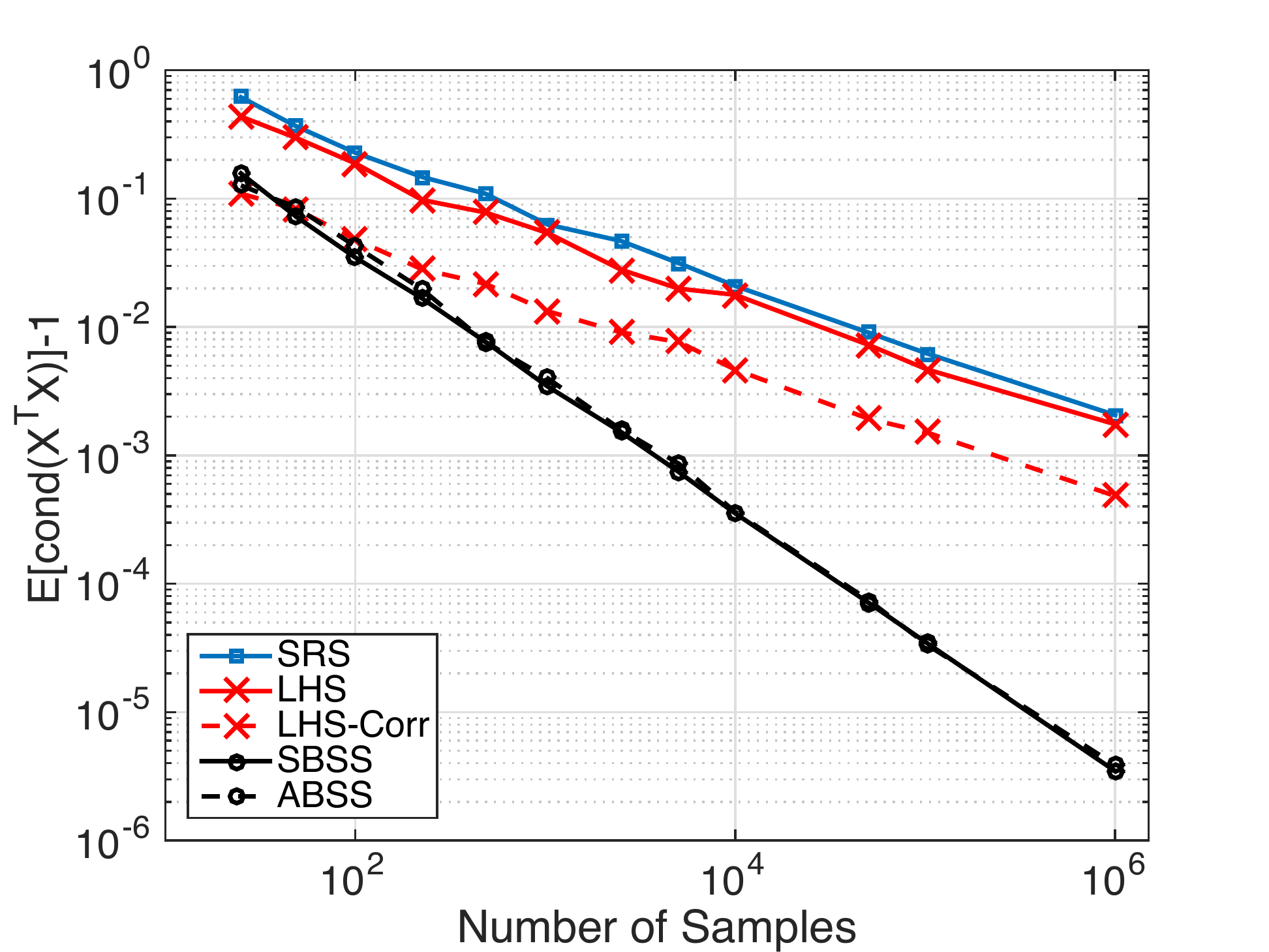}
}
\caption{Evaluation of spurious correlation produced by various 2D sample designs as a function of sample size: (a) standard deviation of spurious correlation, (b) maximum spurious correlation, and (c) condition number.}
\label{fig:spurious_corr} 
\end{figure}
Notice that the stratified samples (both symmetric [SBSS] and asymmetric with 2:1 aspect ratio [ABSS]) produce smaller spurious correlations than even the LHS with explicit correlation control (LHS-Corr) while LHS without correlation control produces relatively large spurious correlations. Moreover, the relative improvement of a stratified design over LHS and LHS-Corr increases with sample size.

Next, we evaluate the condition number of $\mathbf{X}^T\mathbf{X}$, where $\mathbf{X}$ is the $N\times n$ design matrix with values in the probability space scaled to the range $[-1,1]$, defined by:
\begin{linenomath}
\begin{equation}
\text{cond}\left(\mathbf{X}^T\mathbf{X}\right)=\dfrac{\phi_1}{\phi_n}
\end{equation}
\end{linenomath}
where $\phi_1$ and $\phi_n$ are the $1^{st}$ and $n^{th}$ eigenvalues of $\mathbf{X}^T\mathbf{X}$ respectively. The condition number is a more general metric of orthogonality than any single correlation value with a value of $1$ indicating a perfectly orthogonal sample. It is therefore useful for higher dimensional samples.
We consider first the condition number for the 2D samples above as a function of sample size (Figure \ref{fig:6c}) and observe the same general trend. However, as the dimension grows, stratified designs are less effective at achieving orthogonality as demonstrated in Figure \ref{fig:cond_v_D}, which shows the mean condition number from 10 independent trials for samples of size $N=256$ (Figure \ref{fig:7a}) and $N=1024$ (Figure \ref{fig:7b}) as a function of dimension for the different sampling methods. 
\begin{figure}[!ht]
\centering
\subfigure[\label{fig:7a}]{
\centering
\includegraphics[width=0.48\columnwidth]{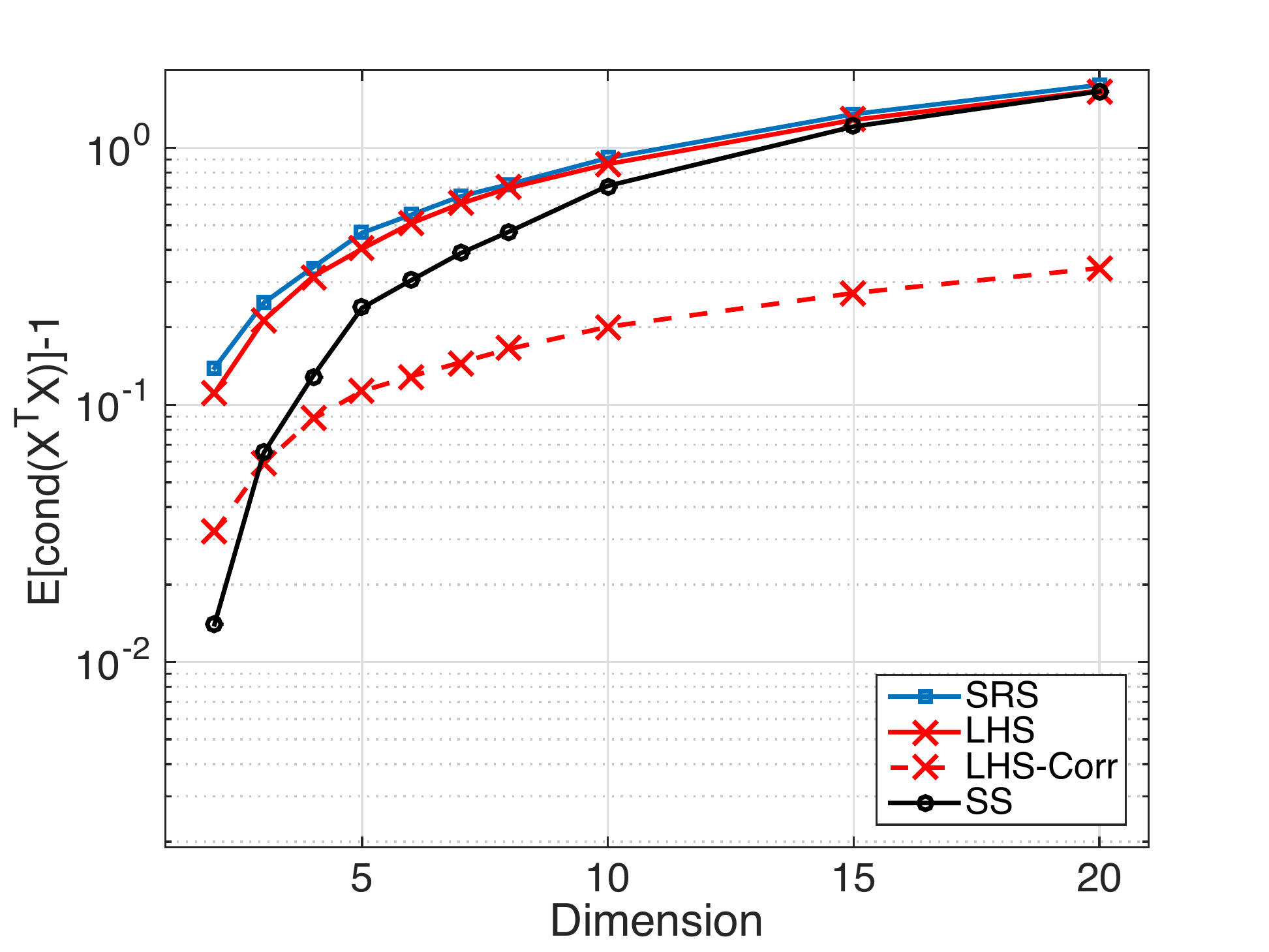}
}
\subfigure[\label{fig:7b}]{
\centering
\includegraphics[width=0.48\columnwidth]{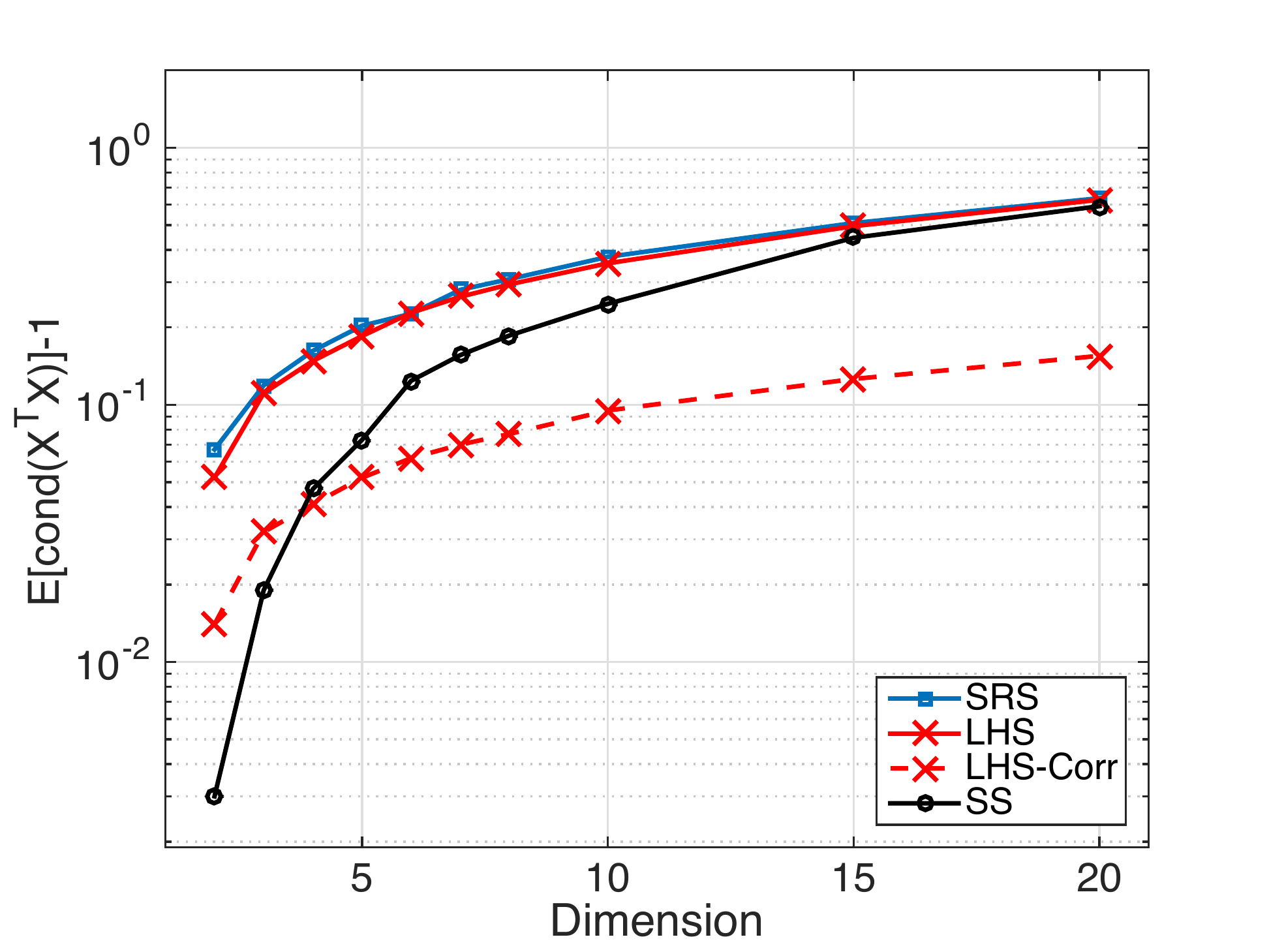}
}
\caption{Condition number of the design matrix as a function of dimension for the different sampling methods from (a) $N=256$ samples and (b) $N=1024$ samples.}
\label{fig:cond_v_D} 
\end{figure}
As the dimension grows, LHS methods with active correlation reduction significantly improves the orthogonality of the sample while SS produces sample sets with reduced spurious correlations up to moderate-dimensional samples without the need for any sample adjustment.
%
%

\subsection{Projective properties and variable interactions}

Often, in an uncertainty analysis, certain variables have a strong effect on the response while others are relatively insignificant. Global, or variance-based, sensitivity analysis is one means of assessing the significance of each variable (and interactions) \cite{Saltelli_et_al_08}. But, this significance can usually only be established {\it{a posteriori}} (i.e.\ after concluding the Monte Carlo study). Consequently, it necessary that a sampling method project all variables through the transformation effectively in order to capture the effects of the significant variables (even when they are not known). This property of a sample design is referred to as the projective property and it is perhaps the greatest strength of LHS. Because LHS discretizes each variable finely (highly resolving all of the marginal distributions), it projects each variable through the transformation well. 

The projective properties of LHS follow directly from the properties illuminated by Stein \cite{Stein_Tech_87}, who showed that LHS has the effect of filtering out the additive components (or main effects) of the transformation. More specifically, Stein showed that if the transformation $h(\mathbf{x})$ is decomposed in terms of its main effects $h_a(\mathbf{x})$ and interaction effects $r(\mathbf{x})$ as:
\begin{linenomath}
\begin{equation}
h(\mathbf{x})=h_a(\mathbf{x})+r(\mathbf{x})
\end{equation}
\end{linenomath}
then the variance of the LHS estimator becomes:
\begin{linenomath}
\begin{equation}
\text{Var}(T_L)=\dfrac{1}{N}\int r(\mathbf{x})^2dF(\mathbf{x})+\mathcal{O}\left(\dfrac{1}{N}\right)
\end{equation}
\end{linenomath}
The variance associated with the main effects $h_a(\mathbf{x})$ is very small, $\mathcal{O}(N^{-1})$, while the variance associated with the interactions $r(\mathbf{x})$ is equal to that of a standard Monte Carlo estimate (i.e.\ there is no variance reduction on the interactions). By contrast, SS reduces the variance on the main effects and the interactions in equal measure. This leads to the following conclusion. LHS has excellent projective properties for individual variables. SS, on the other hand, has moderate projective properties for individual variables but much better projective properties for variable interactions. This benefit, however, diminishes as the dimension grows as demonstrated by the following example.


Consider two simple transformations. The first is an additive function defined by:
\begin{linenomath}
\begin{equation}
Y_1=\dfrac{2}{n}\sum_{i=1}^nX_i.
\end{equation}
\end{linenomath}
where $X_i\sim U(0,1)$. The second is a multiplicative function with strong variable interactions given by:
\begin{linenomath}
\begin{equation}
Y_2=\prod_{i=1}^nX_i.
\label{eqn:prod}
\end{equation}
\end{linenomath}
where $X_i\sim U(1-\sqrt(3),1+\sqrt(3))$. Each transformation has been designed to produce $E[Y_1]=E[Y_2]=1$.  

Figure \ref{fig:add_v_mult} shows the average standard deviation from 1,000 Monte Carlo estimates of $E[Y_1]$ and $E[Y_2]$ from 1024 samples using SRS, LHS, and SS for varying problem dimensions. Note that since $E[X_i]\ne0$, the transformation in Eq.\ \eqref{eqn:prod} possesses both main effects and interaction effects. This can be seen by expressing $X_i$ in terms of zero mean variables $\hat{X}_i$ as $X_i=\hat{X}_i-1$ such that $Y_2=(\hat{X}_1-1)(\hat{X}_2-1)\dots(\hat{X}_n-1)$. It is these main effects that enable a variance reduction from LHS on Eq.\ \eqref{eqn:prod}.
\begin{figure}[!ht]
\centering
\subfigure[\label{fig:8a}]{
\centering
\includegraphics[width=0.48\columnwidth]{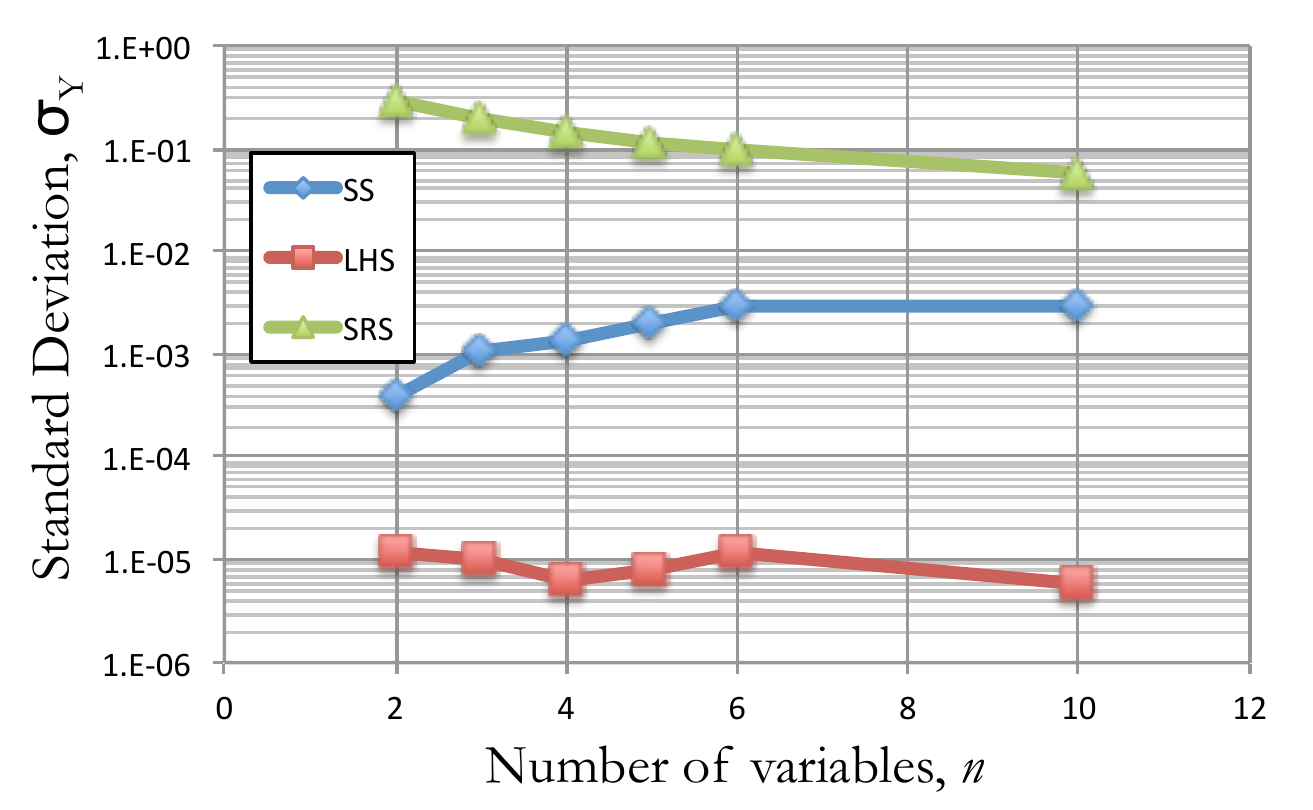}
}
\subfigure[\label{fig:8b}]{
\centering
\includegraphics[width=0.48\columnwidth]{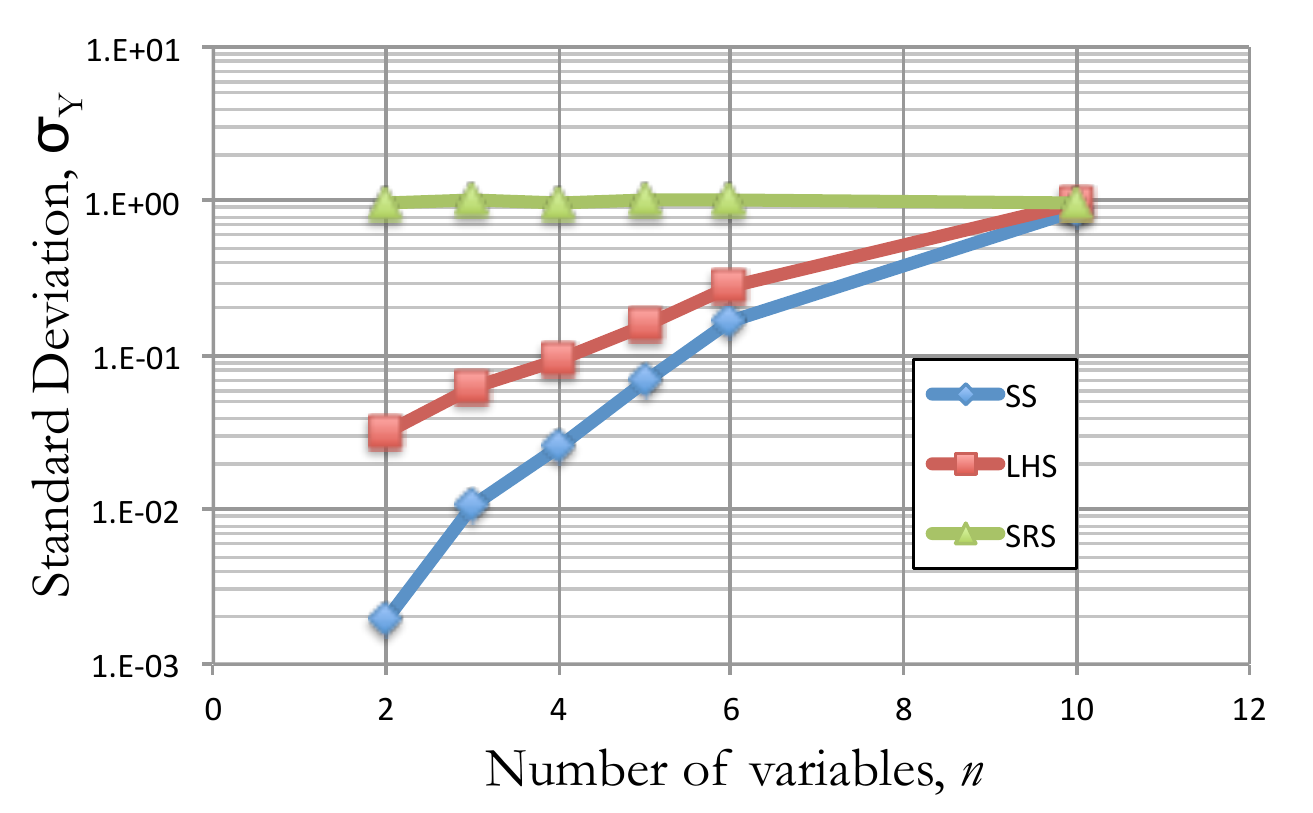}
}
\caption{Response standard deviation for (a) an additive transformation and (b) a multiplicative transformation as a function of dimension for SRS, LHS, and SS.}
\label{fig:add_v_mult} 
\end{figure}
As expected, LHS performs exceptionally well for transformation $Y_1$ regardless of dimension. However, SS reduces variance considerably over LHS for $Y_2$. This variance reduction diminishes with dimension but remains more effective than LHS up to $n=\log_2(N)$ (here $n=10$ for $N=1024$), after which neither method is capable of producing a meaningful variance reduction.


%
%

\section{Sample size extension} 

The adaptive UQ methodology used in this paper requires the ability to easily add samples to an existing set. Sample size extension for SRS is straightforward because samples are iid realizations. However, in SS and LHS, sample size extension is not necessarily trivial. For SS, samples are traditionally added to existing strata although, as we will show, this is not optimal. Extension of LHS meanwhile requires to maintain equal weight for all samples for which a few methodologies have been developed recently. 

The first methodologies developed to extend Latin hypercube samples are referred to as Replicated Latin Hypercubes (RLH) \citep{Iman_Conf_81,McKay_TR_95}. 
RLH entails performing an additional (independent) LHS with identical stratification upon completion of the prior LHS. Using RLH, the minimum achievable sample size extension is limited by the original LHS size (it requires adding $N$ additional samples at each refinement) and the sample does not benefit from any additional space refinement.

Recent works by Tong, \citep{Tong_RESS_06}, Sallaberry et al.\, \citep{Sallaberry_et_al_RESS_08}, and Vorechovsky et al. \citep{Vorechovsky_et_al_ICOSSAR_13} have proposed methods commonly referred to as Hierarchical Latin Hypercube Sampling (HLHS). Although these methods differ in their details (e.g.\ the methods in \citep{Sallaberry_et_al_RESS_08,Vorechovsky_et_al_ICOSSAR_13} enable extension for samples of correlated variables), the basic premise is the same. Given an LHS of size $N$, the strata of each sample component are further divided into $t+1$ strata (one containing the original sample) - where $t$ is referred to as the refinement factor - and the new components are randomly paired as in a typical LHS implementation. 
HLHS introduces $N_{new}=tN$ new samples to the set and in general, through $r$ sample size extensions, the size of the sample set grows exponentially as:
\begin{linenomath}
\begin{equation}
\label{eqn:HLHS_growth}
N_{tot}=N(t+1)^r.
\end{equation}
\end{linenomath}
Its space-filling properties and statistical convergence are improved over RLH though because each sample size extension involves a division of the sample strata.

Lastly, an alternate means of sample size ``extension" is through a method called Nested Latin Hypercube Sampling (NLHS) developed in \citep{Qian_Bio_09} and \citep{Rennen_et_al_SMO_10}. Extension is presented in quotations because it is not truly an extension technique. Rather, it works in the opposite manner as HLHS by producing one very large LHS sample (larger than would presumably be necessary for the problem at hand) and dividing this sample into smaller ``nested" Latin hypercubes. It can therefore be used as an extension method by considering first a single nest and then considering subsequent nests as the sample extensions.

These developments for adding samples to a Latin hypercube design represent a significant development toward achieving the desired adaptive Monte Carlo framework. Each of the methods, with the exception of RLH, strictly maintain the character (and all associated properties) of a Latin hypercube design as discussed in Section 3. Given the broad appeal of LHS and its many desirable features, these methods are attractive for many practical applications.


%
%

\section{Refined Stratified Sampling}

The primary drawback of the LHS extension methods is the rapid sample size growth associated with sample size extension. The methodology developed in this section, referred to as Refined Stratified Sampling (RSS), makes use of the unequal sample weighting allowed by stratified sampling to extend a stratified sample set by a single sample. Furthermore, after each extension, the sample maintains the properties of a stratified sample with weights that can be easily computed if the sample components are independent and uncorrelated. 

Consider the input random vector $\mathbf{X}$ with $n$ independent and uncorrelated components defined on the sample space $\mathbf{\mathcal{S}}$ as described in Section 2. Given an initial stratified sample set of size $N$ distributed over $M=N$ strata, the Refined Stratified Sampling methodology proceeds as follows:
\begin{enumerate}
\item{Select a stratum $\mathbf{\Omega}_k$ to divide according to the following criteria:
\begin{itemize}
\item{If a single stratum $\mathbf{\Omega}_k$ exists such that $w_k>w_j\hspace{6pt}\forall j\ne k$, divide this stratum.}
\item{If $N_s$ strata $\mathbf{\Omega}_k; k=1,2,\dots,N_s$ exist with $w_k=\underset{j}\max{(w_j)}$, randomly select the stratum $\mathbf{\Omega}_k$ to divide with the probability of dividing $\mathbf{\Omega}_k$ equal to $P[D(\mathbf{\Omega}_k)]=\dfrac{1}{N_s}$.}
\end{itemize}
\item{Divide stratum $\mathbf{\Omega}_k$ in half according to the following criteria:
\begin{itemize}
\item{Compute the unstratified lengths of $\mathbf{\Omega}_k$, defined as:
\begin{linenomath}
\begin{equation}
\lambda_{ik}=\xi_{ik}^{hi}-\xi_{ik}^{lo}; i=1,2,\dots,n
\label{eqn:11}
\end{equation}
\end{linenomath}
where $\xi_{ik}^{hi}$ and $\xi_{ik}^{lo}$ are defined as in Section \ref{sec:SS}.} 
\item{Determine the maximum unstratified length of $\mathbf{\Omega}_k$ as $\Lambda_{k}=\underset{i}\max{(\lambda_{ik})}$
}
\item{Divide stratum $\mathbf{\Omega}_k$ along the component $i^*$ corresponding to $\Lambda_{k}$. If $N_c$ components exists such that individual unstratified lengths $\lambda_{i^*k}=\Lambda_k; i^*=1,2,\dots,N_c$ then randomly select the component to divide with the probability of dividing component $i^*$ equal to $P[D(i^*)]=\dfrac{1}{N_c}$.}
\end{itemize}
}

}
\item{Keeping all existing samples $\mathbf{x}_l; l=1,2,\dots,N$, randomly sample in the newly defined empty stratum. Compute new weights for each of the existing samples and the new sample according to Eq. \eqref{eqn:3}.}
\item{Repeat steps 1 - 3 for every new extension.}
\end{enumerate}

The RSS process is described graphically in the flowchart provided in Figure \ref{fig:3} for a two-component random vector starting with $N=M=1$ and showing the first four sample-size extensions using RSS.
\begin{figure}
\centering
\includegraphics[width=1.\columnwidth]{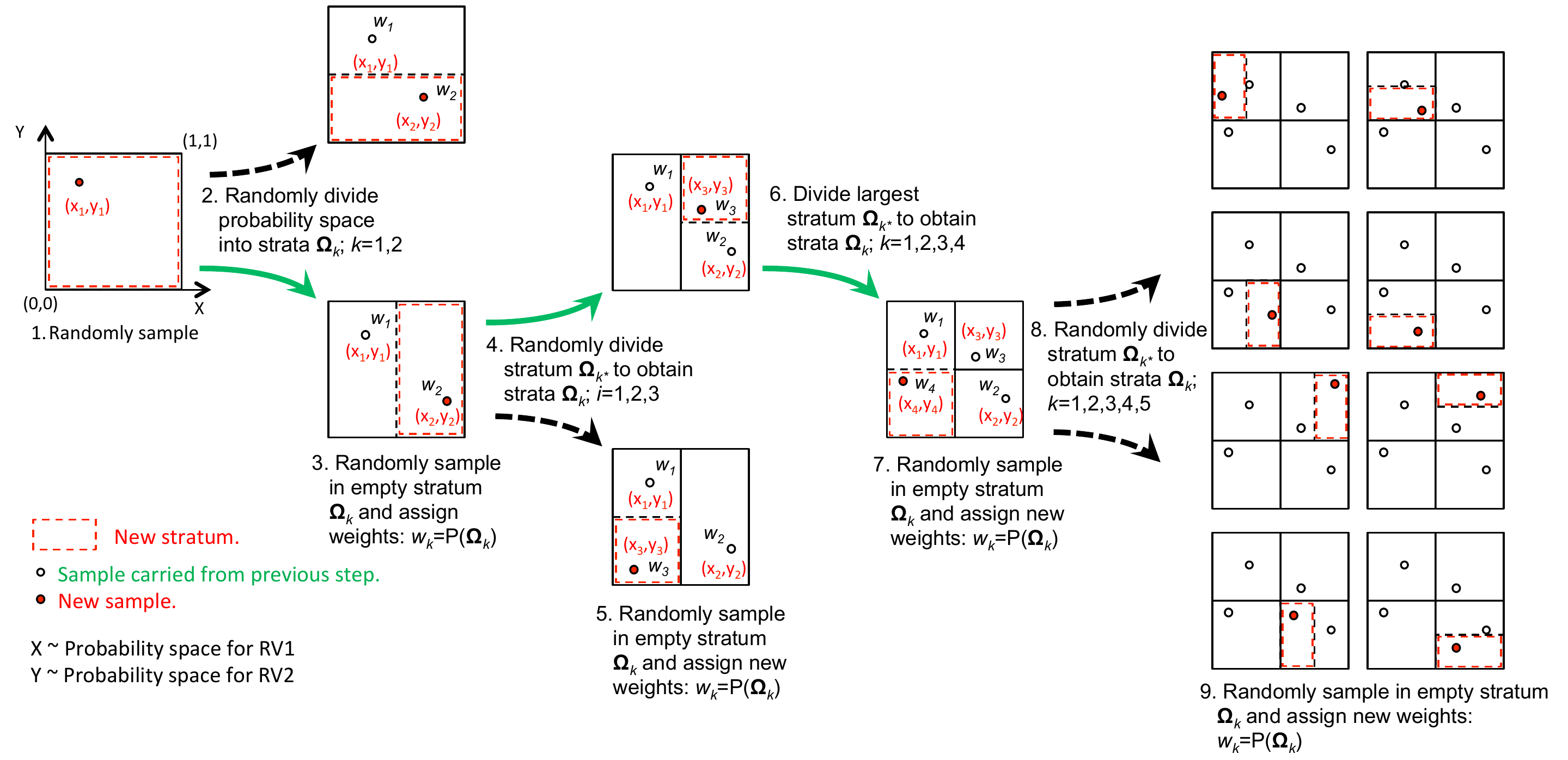}
\caption{Refined Stratified Sampling: Flowchart of sample size extension procedure for two random variables.}
\label{fig:3}
\end{figure}
\clearpage


\subsection{Why refine strata?}
The methodology presented herein utilizes a refinement of the strata definitions rather than adding samples to existing strata. The benefits of this strategy can be significant - but can also be detrimental if not implemented appropriately. The benefits of the method rely on the following theorem.

\begin{theorem}
\label{thrm:restratification}
Let $\mathbf{\Omega}_1$ denote a stratum of space $\mathbf{\mathcal{S}}$ and $\mathbf{\omega}_k;\hspace{3pt}k=1,\dots,N_{ss}$ denote $N_{ss}$ disjoint substrata of $\mathbf{\Omega}_1$ such that $\cup_{k=1}^N\mathbf{\omega}_k=\mathbf{\Omega}_1$ and $P[\omega_k]=P[\omega_j] \hspace{3pt} \forall k,j \in 1,\dots,N_{ss}$. The variance of a statistical estimator computed using $N_{ss}$ samples drawn singularly from $\mathbf{\omega}_k;\hspace{3pt}k=1,\dots,N_{ss}$ will always be less than the variance of the same estimator computed using $N_{ss}$ samples drawn from $\mathbf{\Omega}_1$.
\end{theorem}

\begin{proof}
Consider a statistical estimator from a stratified design, $T_S$, as given in Section \ref{sec:stat_eval}. Next, consider two different stratified designs wherein all strata are identical except in one region, denoted $\mathbf{\Omega}_1$. In design one, the region $\mathbf{\Omega}_1$ has only one stratum possessing $N_{ss}$ samples. In design two, $\mathbf{\Omega}_1$ is divided into $N_{ss}$ balanced strata, each possessing a single sample. The variances of the statistical estimators from these two designs are given by:
\begin{linenomath}
\begin{subequations}
\begin{align}
\text{Var}\left[T_{s1}\right]&=\dfrac{p_1^2}{N_{ss}}\sigma_1^2+\sum_{j=2}^M\dfrac{p_j^2}{N_j}\sigma_j^2\\
\text{Var}\left[T_{s2}\right]&=\sum_{k=1}^{N_{ss}}p_k^2\sigma_k^2+\sum_{j=2}^M\dfrac{p_j^2}{N_j}\sigma_j^2
\end{align}
\end{subequations}
\end{linenomath}
where $\sum_{k=1}^{N_{ss}}p_k=p_1$ and the $M$ term summation refers to all points in the stratified design outside of region $\mathbf{\Omega}_1$. The difference in variance between these estimators is given by:
\begin{linenomath}
\begin{equation}
\text{Var}\left[T_{s1}\right]-\text{Var}\left[T_{s2}\right]=\dfrac{p_1^2}{N_{ss}}\sigma_1^2-\sum_{k=1}^{N_{ss}}p_k^2\sigma_k^2
\end{equation}
\end{linenomath}
Under the condition that $p_k=\dfrac{p_1}{N_{ss}}$ (i.e. balanced stratum refinement of $\mathbf{\Omega}_1$), the difference reduces to:
\begin{linenomath}
\begin{equation}
\text{Var}\left[T_{s1}\right]-\text{Var}\left[T_{s2}\right]=\dfrac{p_1^2}{N_{ss}}\left(\sigma_1^2-\dfrac{1}{N_{ss}}\sum_{k=1}^{N_{ss}}\sigma_k^2\right)
\label{eqn:difference}
\end{equation}
\end{linenomath}
It is straightforward to show that the term on the right hand side of Eq.\ \eqref{eqn:difference} is strictly positive. Therefore, variance is reduced using a balanced restratification of $\mathbf{\Omega}_1$.\qed
\end{proof}


\subsubsection{Optimal strata refinement for various output distributions}

Theorem \ref{thrm:restratification} states that a balanced stratum refinement of any given stratum will \emph{always} reduce the variance of statistical estimates when compared with simply adding samples to the existing stratum. Note, however, that Theorem \ref{thrm:restratification} does not say that any stratum refinement in general will result in reduced variance. In fact, improper stratum refinement can increase the variance. Furthermore, Theorem \ref{thrm:restratification} does not imply that a balanced stratum refinement produces the optimal variance reduction. It often will not.

Consider a statistical estimate of the expected value of $Y$, $T_S=E\left[Y\right]$, from a stratified design. Here, we establish the optimal stratum division (on the basis of variance reduction) for different distributions of $Y$ produced from the transformation $Y=F(\mathbf{X})$. For demonstration purposes, the estimate is computed initially from two samples, one in each stratum $\mathbf{\Omega}_1$ and $\mathbf{\Omega}_2$ with $p_1=P\left[\mathbf{\Omega}_1\right]=0.5$ and $p_2=P\left[\mathbf{\Omega}_2\right]=0.5$ as depicted in Figure \ref{fig:opt_RSS_norm}.
\begin{figure}
\centering
\includegraphics[width=0.5\columnwidth]{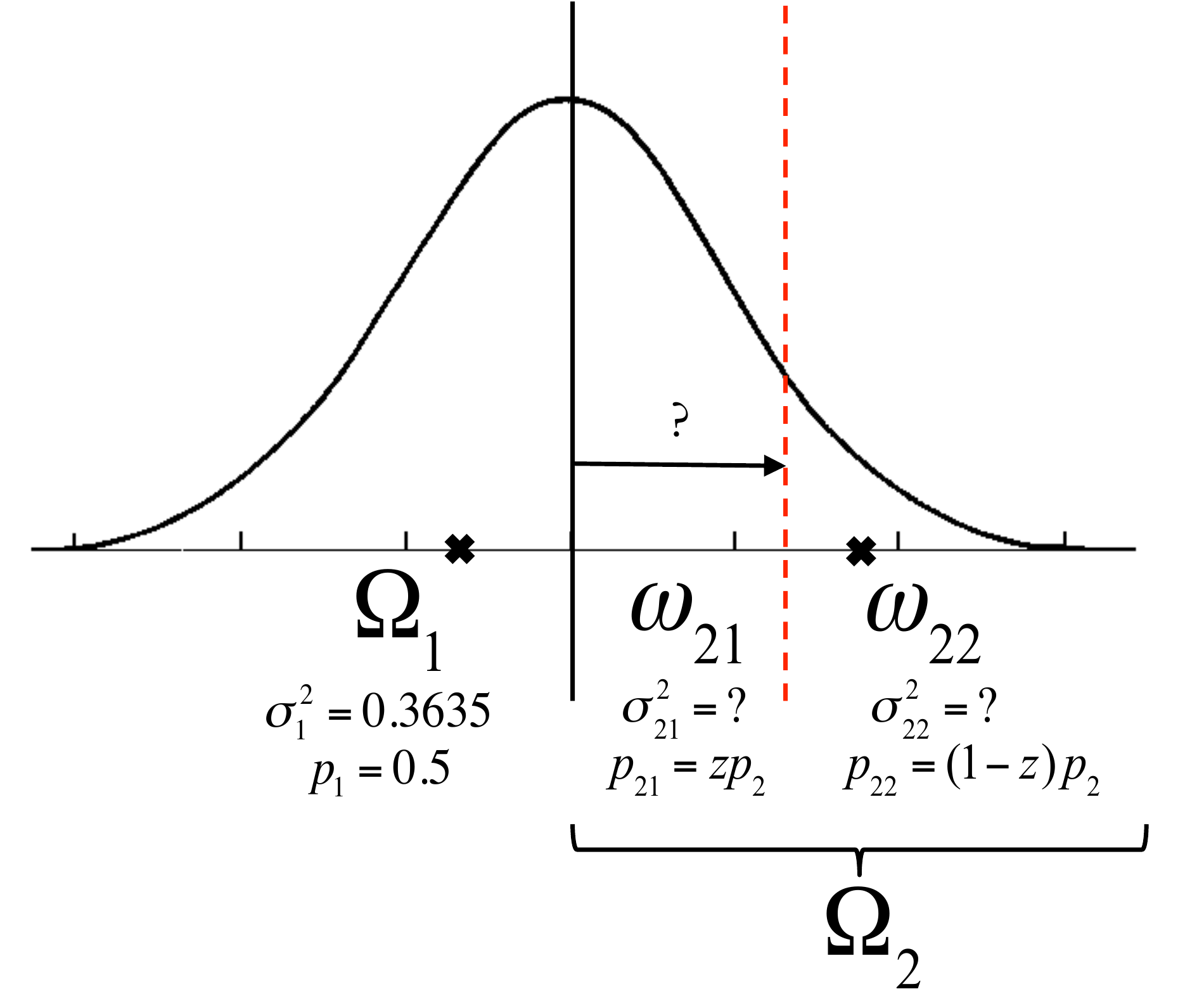}
\caption{Identification of the optimal stratum refinement location for a normal distribution.}
\label{fig:opt_RSS_norm}
\end{figure}
We are interested in identifying the optimal location to divide the strata such that a single sample can be added while minimizing $\text{Var}\left[T_S\right]$. In each example, stratum $\mathbf{\Omega}_2$ is divided into two substrata denoted $\omega_{21}$ and $\omega_{22}$ possessing probability weights $p_{21}$ and $p_{22}$ subject to:
\begin{linenomath}
\begin{subequations}
\begin{align}
p_{21}&=P\left[\mathbf{\omega}_{21}\right]=zp_2\\
p_{22}&=P\left[\mathbf{\omega}_{22}\right]=(1-z)p_2
\end{align}
\end{subequations}
\end{linenomath}
where $z\in[0,1]$, referred to as the `Imbalance Factor,' is a measure of imbalance of the resulting stratum division (Figure \ref{fig:opt_RSS_norm}). Note that $z=0.5$ corresponds to balanced stratum refinement. The objective of the optimization can be stated as:
\begin{linenomath}
\begin{equation}
	\begin{aligned} 
	& \underset{z}{\text{minimize:}}
	& & \text{Var}\left[T_s\right]=p_1^2\sigma_1^2+p_{21}^2\sigma_{21}^2+p_{22}^2\sigma_{22}^2\\
	& \text{subject to:}
	& & p_{21}=zp_{2} \\
	& & & p_1+p_{21}+p_{22} = 1
	\end{aligned}
	\label{eq:Error}
\end{equation}
\end{linenomath}

We consider three different output distributions: $Y\sim \text{Normal}(0,1)$; $Y\sim \text{Uniform}(0,1)$, $Y\sim \text{LogNormal}(-1.49,1.27)$. Optimization as defined by Eq.\ \eqref{eq:Error} yields the stratum refinement outlined in Table \ref{tab:opt_division} for each distribution.
\begin{table}
\centering
\caption{Optimal strata division for uniform, normal, and lognormal output distributions.}
\begin{tabular}{llllll}
\hline
Dist. & $p_{21}$ & $p_{22}$ & $\sigma_{21}^2$ & $\sigma_{22}^2$ & $z$ \\\hline
$U(0,1)$ & 0.25 & 0.25 & $5.2e-3$ & $5.2e-3$ & 0.5 \\
$N(0,1)$ & 0.3162 & 0.1838 & $6.53e-2$ & $0.21$ & 0.632 \\
$LN(-1.49,1.27)$ & 0.4426 & 0.0574 & $0.1165$ & $6.97$ & 0.885 \\
\hline
\end{tabular}
\label{tab:opt_division}
\end{table}
Table \ref{tab:var_divisions} meanwhile compares the resulting variances of $T_S$ for each distribution considering different strata definitions: 1. Two samples - one from $\mathbf{\Omega}_1$ and one from $\mathbf{\Omega}_2$; 2. Three samples with no refinement - a single sample is added to stratum $\mathbf{\Omega}_2$; 3. Three samples with balanced refinement - stratum $\mathbf{\Omega}_2$ is divided such that $p_{21}=p_{22}=0.25$; 4. Three samples with optimal refinement of stratum $\mathbf{\Omega}_2$ as defined in Table \ref{tab:opt_division}.
\begin{table}
\centering
\caption{Variance of the estimated mean value for different strata refinements and different output distributions.}
\begin{tabular}{lllll}
\hline
 \multicolumn{5}{c}{$\text{Var}[T_S]$}\\\hline
 & & Three Samples & Three Samples & Three Samples \\
Dist. & Two Samples & No Refinement & Bal. Refinement & Opt. Refinement \\\hline
$U(0,1)$ & 0.0104 & $7.8125e-3$ & $5.859e-3$ & $5.859e-3$ \\
$N(0,1)$ & 0.18175 & 0.1363 & 0.1083 & 0.1045  \\
$LN(-1.49,1.27)$ & 0.4138 & 0.2073 & $0.1696$ & $0.04667$  \\
\hline
\end{tabular}
\label{tab:var_divisions}
\end{table}
Notice that optimal stratification can have a profound influence on variance reduction.

Figure \ref{fig:imbalance_factor} shows the variance of the mean estimates for the normal, uniform, and lognormal distributions as a function of the imbalance factor $z$ with the optimal value demarcated by an `x'. 
\begin{figure}
\centering
\subfigure[\label{fig:10a}]{
\centering
\includegraphics[width=0.31\columnwidth]{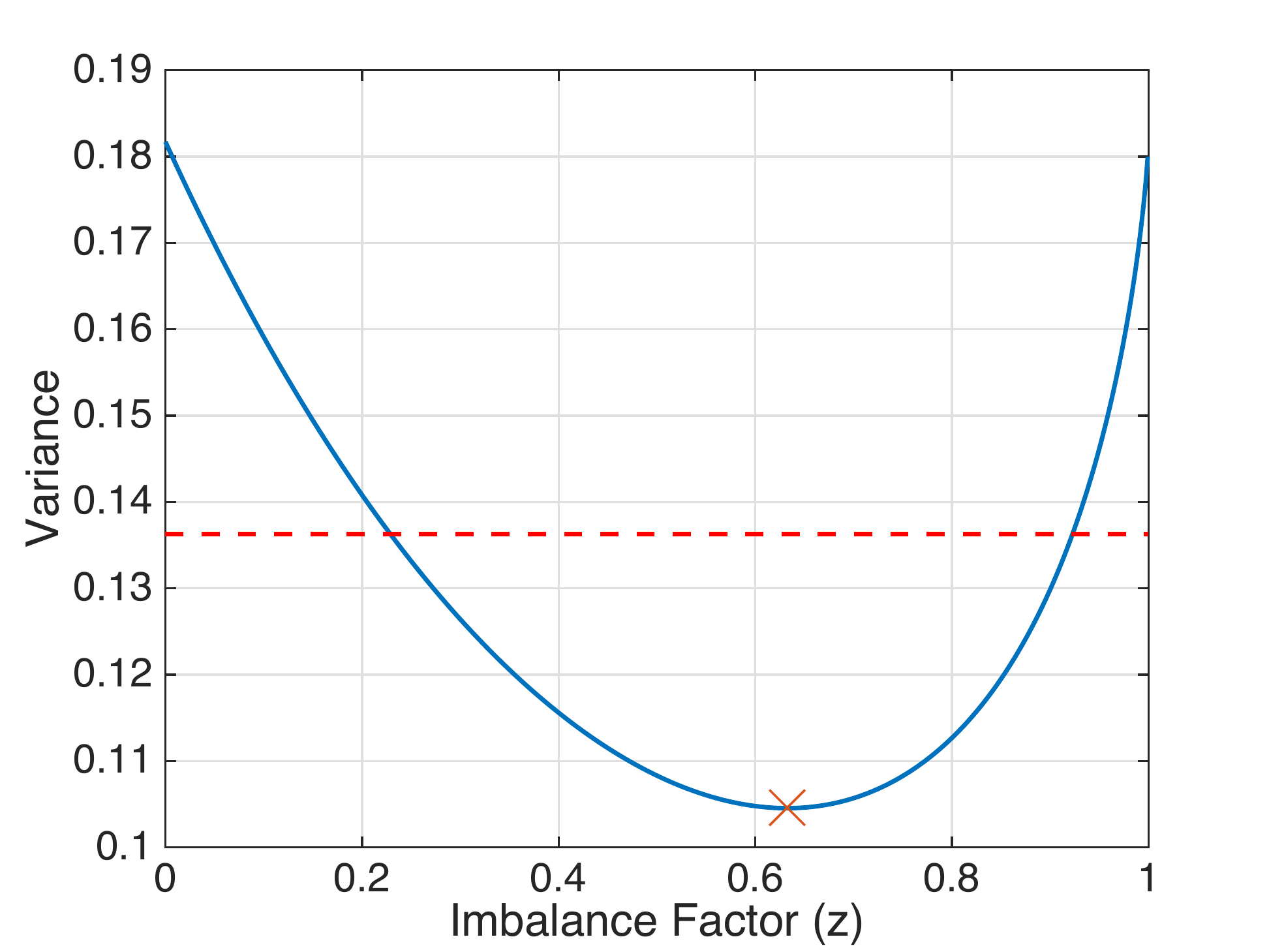}
}
\subfigure[\label{fig:10b}]{
\centering
\includegraphics[width=0.31\columnwidth]{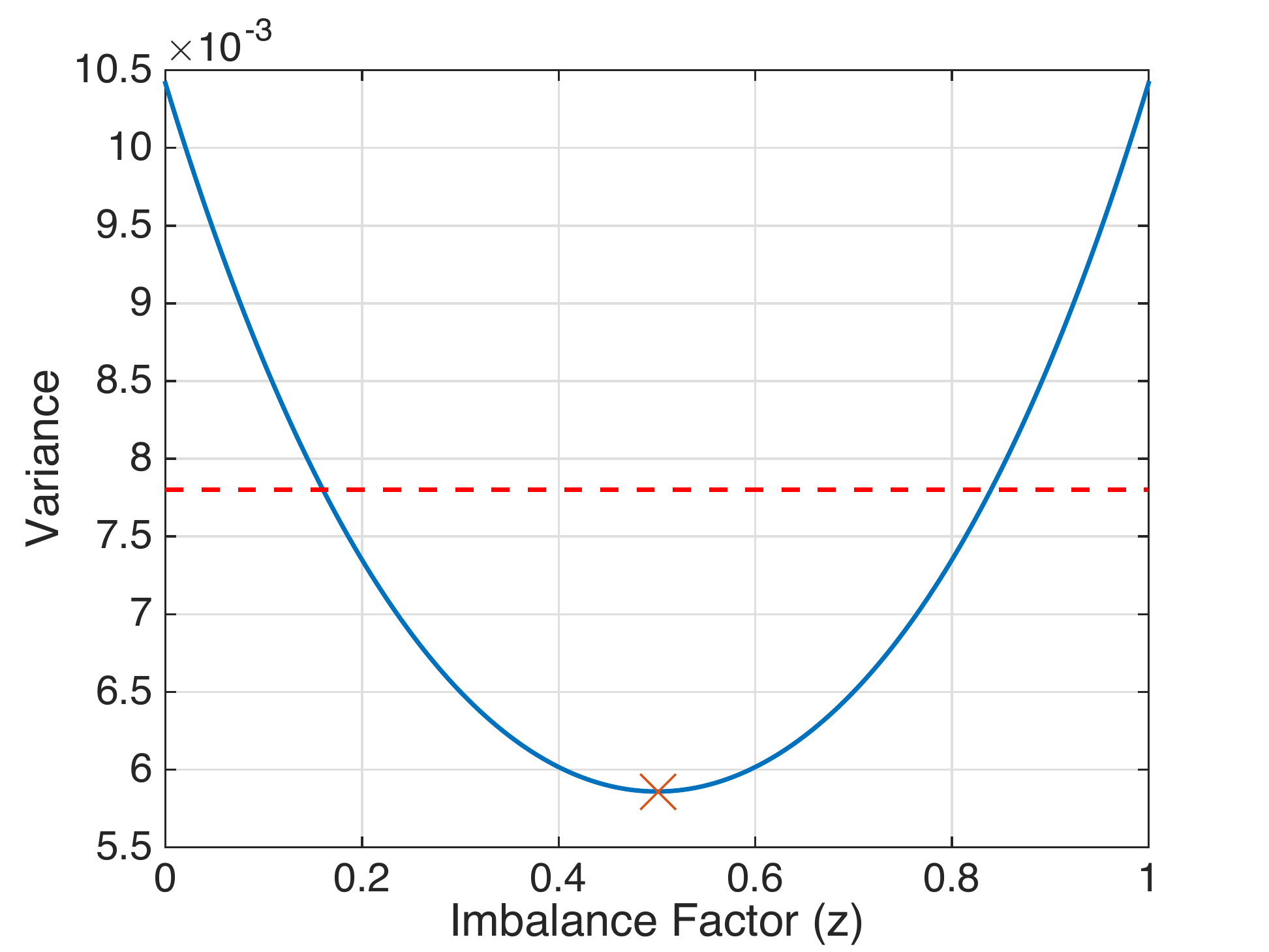}
}
\subfigure[\label{fig:10c}]{
\centering
\includegraphics[width=0.31\columnwidth]{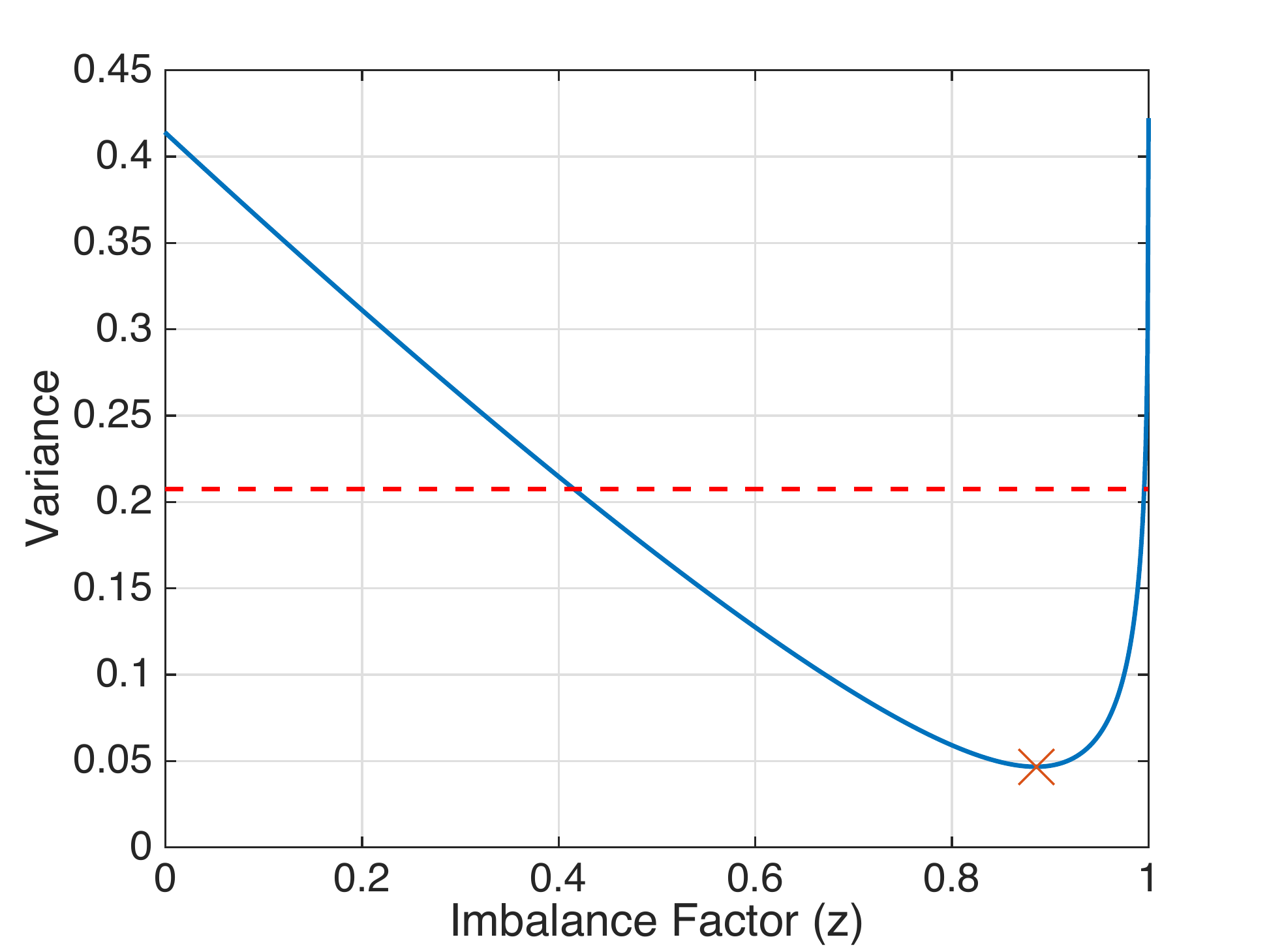}
}
\caption{Variance of the estimated mean value from three samples with (a) normal, (b) uniform, and (c) lognormal output distributions as a function of stratum imbalance factor.}
\label{fig:imbalance_factor}
\end{figure}
The dashed line shows the variance resulting from samples produced without stratum refinement. Notice that the balanced refinement is not necessarily the optimal but always reduces variance while many refinement strategies (e.g.\ $z<0.4$ for lognormal output) actually increase the variance relative to estimates produced without refinement. Analysis of these results shows that the closer the stratum $\mathbf{\Omega}_2$ is to possessing symmetric conditional distribution $D_Y(y|y\in \mathbf{\Omega}_2)$, the closer the balanced refinement strategy becomes to optimal. Conversely, when the conditional distribution $D_Y(y|y\in \mathbf{\Omega}_2)$ possesses strong asymmetry, the optimal refinement will be unbalanced.
In such cases, the balanced refinement procedure proposed herein produces diminished return in terms of variance reduction. 

The analysis performed in this section assumes stratification in the domain of the output $Y$. In general, the corresponding strata in the input domain of $\mathbf{X}$ cannot be identified unless $F(\mathbf{X})$ is known and invertible. Additionally, the distribution of $Y$ is not known in general. Consequently, it can be difficult and sometimes impossible to identify the stratum division for $\mathbf{X}$ that corresponds to the optimal division on $Y$ - particularly if $F(\mathbf{X})$ is strongly nonlinear, non-monotonic, discontinuous, or otherwise complex. Balanced stratum refinement on the input space, on the other hand, necessarily corresponds to a balanced stratum refinement on the output space and will always provide a reduction in variance. For this reason, its use is recommended when stratum optimization is not possible. 

Identifying optimal stratum refinement, however, represents a potentially important direction for future research. While some similar concepts have been developed for optimally sampling the input space with probabilistically weighted draws to match the input probability model (e.g.\ \cite{Grigoriu_AMM_09}) the potential capability to use RSS within the adaptive framework promoted herein represents an opportunity to define a partitioning of the input space that optimizes the partitioning of the output space. This development is one that is potentially very powerful for efficient uncertainty analysis. An initial exploration along these lines is undertaken in a parallel effort for applications in reliability analysis \cite{Shields_Sundar_IJRS_2015} where strata in the input variable space are divided such that samples are concentrated in the output space in the region of the limit state function. 


\subsection{RSS yields unbalanced stratified designs}

Globally, RSS produces stratified designs that are balanced only when $N=M=2^p$, are symmetrically balanced only when $N=M=2^{np}$ for integer values of $p$, and are otherwise unbalanced (assuming an initial sample size of one). In general, producing unbalanced designs is recommended only when the output distribution calls for it (see the previous section). However, the output distribution is unknown in general and therefore balanced stratification should be used whenever optimization is not possible. For this reason, it is recommended that the initial stratified design be a balanced one (asymmetric or symmetric). The proposed methodology will then produce unbalanced designs during most iterations but this is justified by Theorem 1.
That said, the analysis conducted in the previous section and the variance reduction expression in Eq.\ \eqref{eqn:8} together imply that optimal refinement will not rely on randomly selecting the stratum to divide but rather selecting it based on some established criteria that may optimize variance reduction. One condition that could be considered, for example, is to enforce some symmetry conditions in stratum selection such that the stratum design does not become too heavily imbalanced. 


\subsection{Expected performance and extension to high dimensional applications}

As demonstrated in Section 5.1, the efficiency of the method is dependent on the resulting distribution (and more specifically the conditional distributions of the response from within each stratum). Given these considerations, the method is expected to converge more rapidly for conditional output distributions that are closer to symmetric. For problems whose output distribution is strongly asymmetric and characterized by a dominant peak and heavy tails (i.e. those with high skewness and kurtosis), the convergence rate of the RSS method will be reduced due to the sub-optimality of the stratum refinement. Optimal stratum refinement methods need to be developed to fully capitalize on the variance reduction potential of RSS.

In its present form, the proposed methodology is not well-suited to problems with very high dimensional input vectors. As the analysis in Section 3 implies, stratified sampling loses considerable effectiveness for dimensions above $n=\log_2(N)$ and its benefits are most pronounced only for low dimensional problems. To extend the RSS method to very high dimensional spaces, it is currently being considered to conduct RSS on low-dimensional subspaces with samples generated on these subspaces randomly paired as in LHS. Moreover, a major deficiency of the proposed method from a dimensionality perspective is that rectilinear stratification requires stratum division in only a single dimension. A generalized RSS algorithm would refine strata of arbitrary shape (i.e.\ Voronoi cells) to allow refinement in multiple dimensions simultaneously - thereby greatly reducing dimensional dependence. Such algorithms are the topic of a forthcoming work.

\section{Applications}


\subsection{Statistical analysis of a stochastic function}
Consider the cubic polynomial function given by:
\begin{linenomath}
\begin{equation}
Y=F(\mathbf{X})=X_1^2X_2-\alpha X_1X_2^2+X_1X_2
\label{eqn:cubic_function}
\end{equation}
\end{linenomath}
The function possesses three random variables: $X_1$, $X_2$, and $\alpha$. Several sets of input distributions are considered to elicit significantly different response distributions as described through analytically computed moments in Table \ref{tab:1}. The cases cover a vast range of distributions from those approaching uniform to those with sharp peaks and heavy tails.
Convergence is defined based on the following relative error:
\begin{linenomath}
\begin{equation}
\left|\dfrac{\sigma_y^2-\text{Var}[Y]}{\text{Var}[Y]}\right|\le\epsilon_{th}
\label{eqn:10}
\end{equation}
\end{linenomath}
where $\text{Var}[Y]$ is the analytically determined variance, $\sigma^2_y$ is the variance computed from the generated samples, and $\epsilon_{th}=0.01$ (in general, analytical moments will not be available - this example is selected for demonstration purposes to allow a simple and clearly defined convergence criterion). 

\begin{table}
\caption{Input distributions and resulting analytical moments of the output. [LN = LogNormal($\mu,\sigma$), U = Uniform($a,b$), N = Normal($\mu,\sigma$)].}
\centering
\begin{tabular}{llllllll}
\hline
Dist. & $X_1$ & $X_2$ & $\alpha$ & $E[Y]$ & Var[$Y$] & Skew[$Y$] & Kurt[$Y$]\\\hline
A & LN(0,0.01) & U(0,20) & N(1,0.1) & -113.33 & 12012.0 & -0.77 & -0.55\\
B & LN(0,0.1) & U(0,10) & N(1,0.1) & -23.37 & 621.18 & -0.89 & -0.24\\
C & LN(0,0.1) & U(0,7) & N(1,0.1) & -9.32 & 121.97 & -0.97 & -0.09\\
D & LN(0,0.1) & U(0,6) & N(1,0.1) & -5.98 & 58.46 & -1.01 & 0.002\\
E & LN(0,0.1) & U(0,5) & N(1,0.1) & -3.31 & 23.65 & -1.08 & 0.17\\
F & LN(0,0.3) & U(0,5) & N(1,0.1) & -3.10 & 25.03 & -1.17 & 0.81\\
G & LN(0,0.4) & U(0,5) & N(1,0.1) & -2.87 & 26.80 & -0.99 & 2.08\\
H & LN(0,0.45) & U(0,5) & N(1,0.1) & -2.70 & 28.85 & -0.48 & 9.42\\
I & LN(0,0.475) & U(0,5) & N(1,0.1) & -2.60 & 30.56 & 0.09 & 23.84\\
J & LN(0,0.5) & U(0,5) & N(1,0.1) & -2.48 & 33.05 & 1.08 & 60.62\\
\hline
\end{tabular}
\label{tab:1}
\end{table}

For each of the parameter sets, 1,000 calculation sets were performed using RSS, HLHS, and SRS. Each calculation set begins with 20 samples and is extended until the convergence criterion is satisfied. The number of samples required for convergence is recorded for each set. The results of this study are summarized in Figure \ref{fig:comparison}.
\begin{figure}
\centering
\subfigure[\label{fig:11a}]{
\centering
\includegraphics[width=0.48\columnwidth]{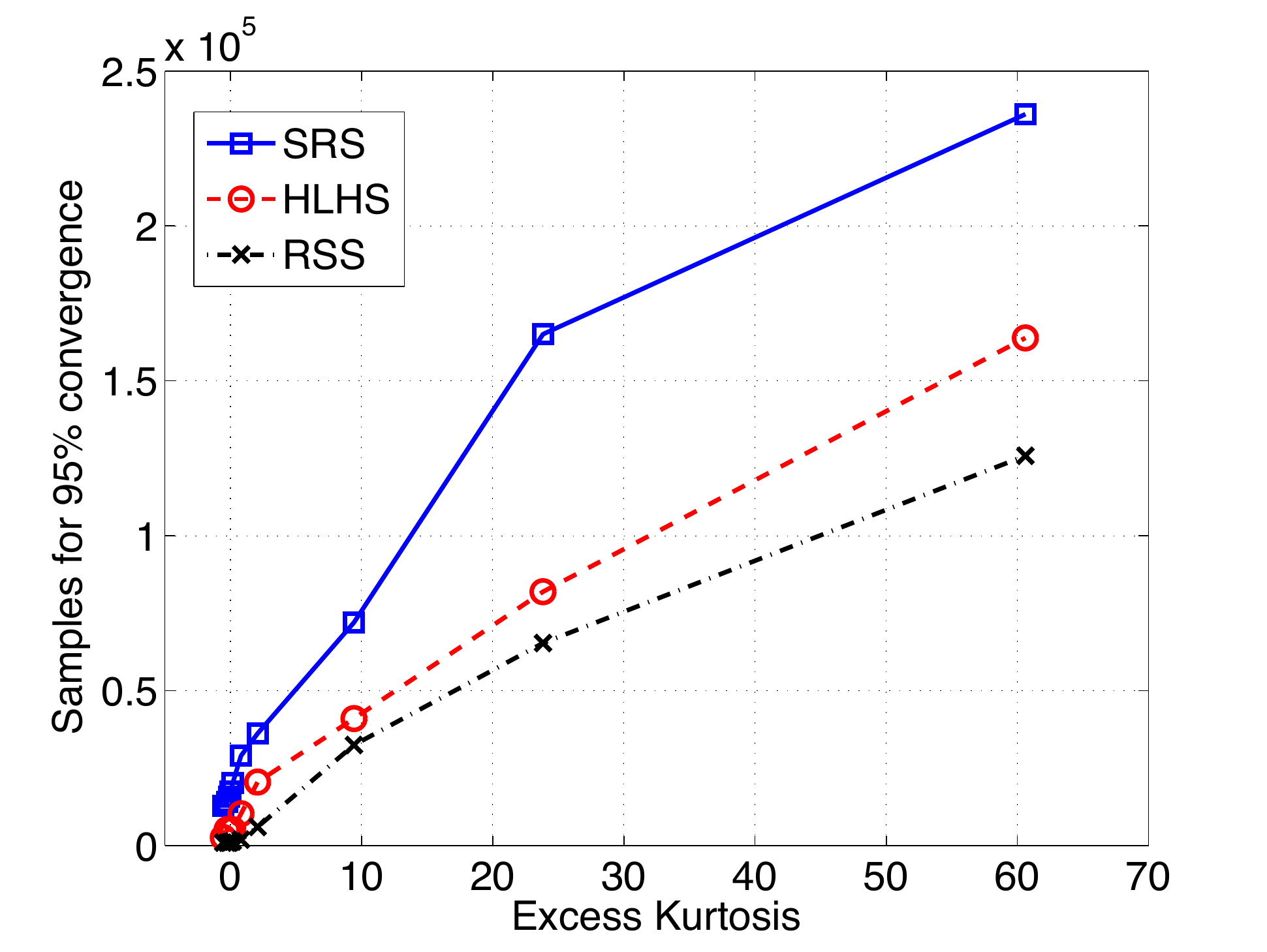}
}
\subfigure[\label{fig:11b}]{
\centering
\includegraphics[width=0.48\columnwidth]{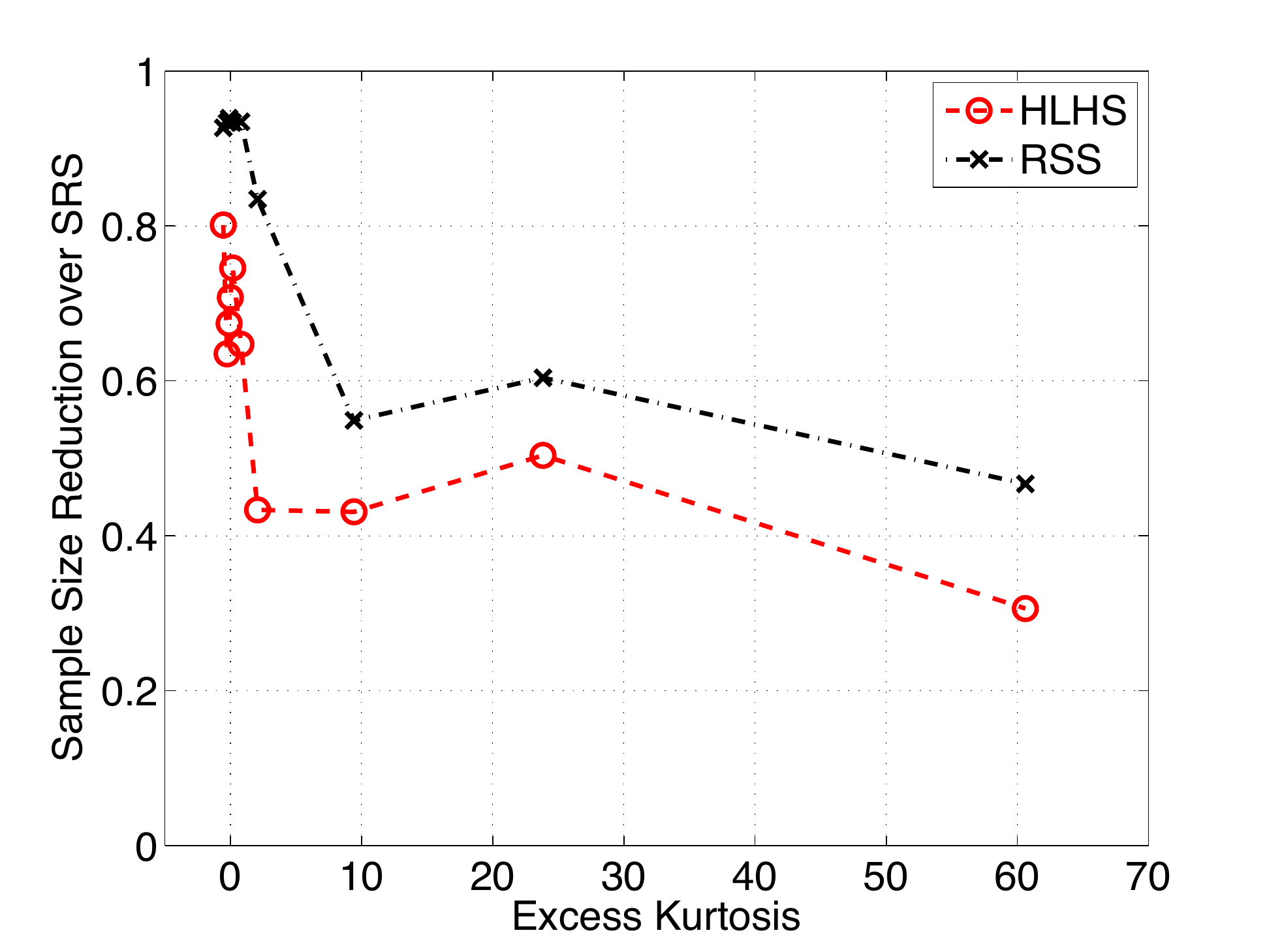}
}
\caption{(a) Number of sample required to obtain 95\% convergence using SRS, HLHS, and RSS as a function of response kurtosis. (b) Reduction in sample size for HLHS and RSS compared to SRS as a function of response kurtosis.}
\label{fig:comparison}
\end{figure}
Figure \ref{fig:11a} shows the number of samples required to achieve convergence for 95\% of the calculation sets using SRS, HLHS, and RSS as a function of response excess kurtosis. In general, a larger number of samples are required for convergence as the kurtosis increases although the proposed RSS method consistently outperforms both SRS and HLHS in this regard. Figure \ref{fig:11b} shows the relative reduction in sample size over SRS for both HLHS and RSS defined by:
\begin{linenomath}
\begin{equation}
\dfrac{N_{S}-N_{H/R}}{N_{S}}
\end{equation}
\end{linenomath}
where $N_S$, $N_{H/R}$ are the number of samples required for 95\% convergence (plotted on the left) for SRS and HLHS/RSS respectively. By this measure, the RSS method affords upwards of a 90\% reduction in sample size when compared to SRS for responses with low kurtosis. Initially, the effectiveness drops precipitously as the kurtosis increases. After kurtosis $\approx$ 2, it stabilizes while still providing significant sample size reduction. Notice also that RSS consistently produces a 10-20\% greater reduction in sample size than HLHS. A complete tabular summary of this study is provided in the Appendix.


Next, we consider the convergence of the empirical CDFs. The ``true" CDFs $D_Y(y)$ for each case A-J are determined by Monte Carlo simulation with 100,000 stratified samples. Convergence to $D_Y(y)$ is measured for the first 10,000 samples using the area validation metric (a.k.a.\ the Minkowski $L_1$ norm) as \citep{Roy_Oberkampf_CMAME_11}:
\begin{linenomath}
\begin{equation}
\delta(l)=\int_{-\infty}^{\infty}|D_Y(y)-\hat{D}_Y^{(l)}(y)|dy.
\label{eqn:14}
\end{equation} 
\end{linenomath}
where $\hat{D}_Y^{(l)}(y)$ denotes the empirical CDF computed from $l$ samples. Sample plots showing the convergence of $\delta(l)$ using SRS, HLHS, and RSS are provided in Figure \ref{fig:5}. While HLHS and RSS provide comparable rates of convergence, HLHS allows only those sample sizes shown with an `x'; highlighting the need to add many samples during an HLHS sample size extension. RSS meanwhile affords extension by as many or as few samples as desirable for the application. 
\begin{figure}[H]
\centering
\includegraphics[width=0.49\columnwidth]{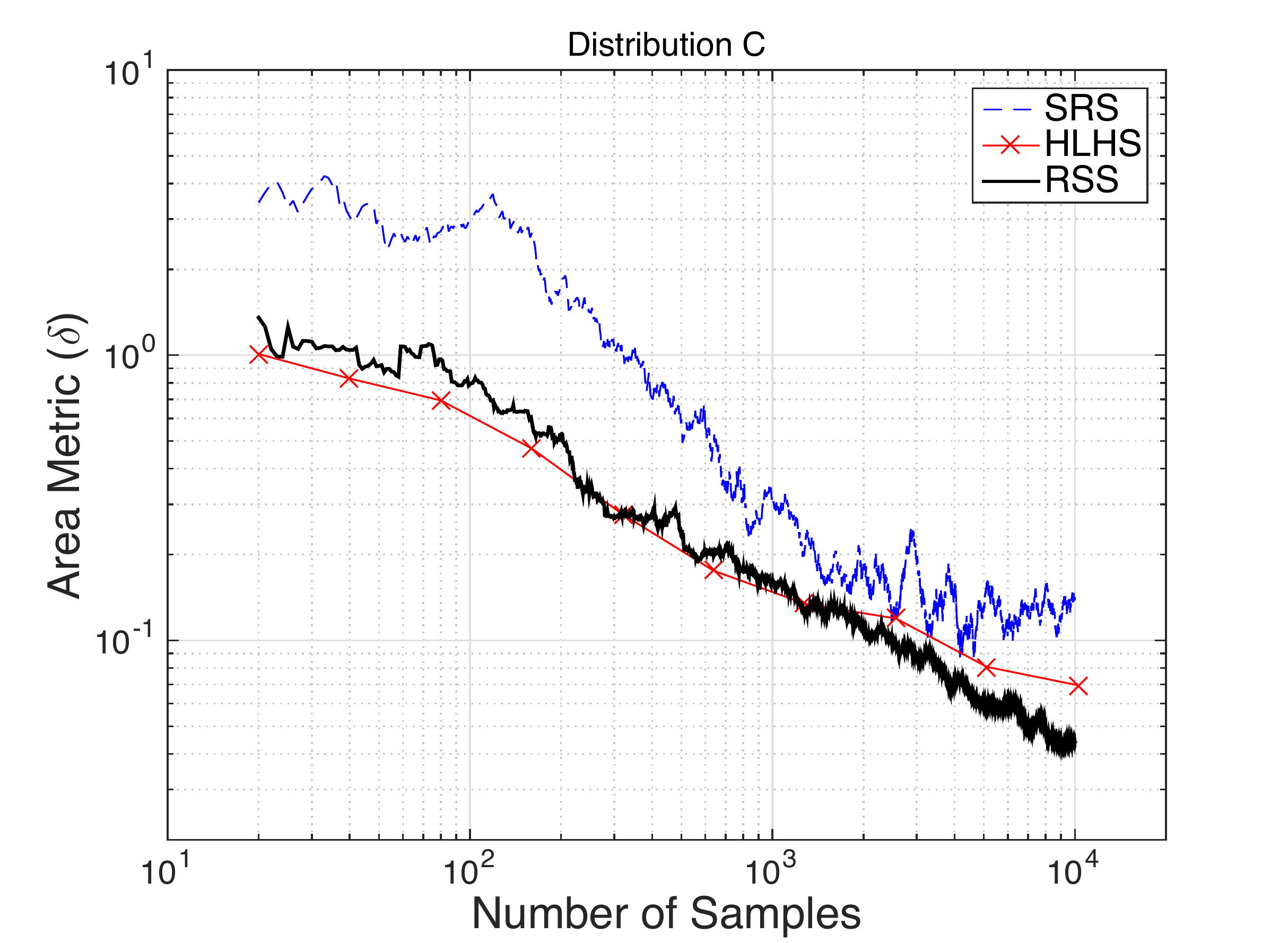}
\includegraphics[width=0.49\columnwidth]{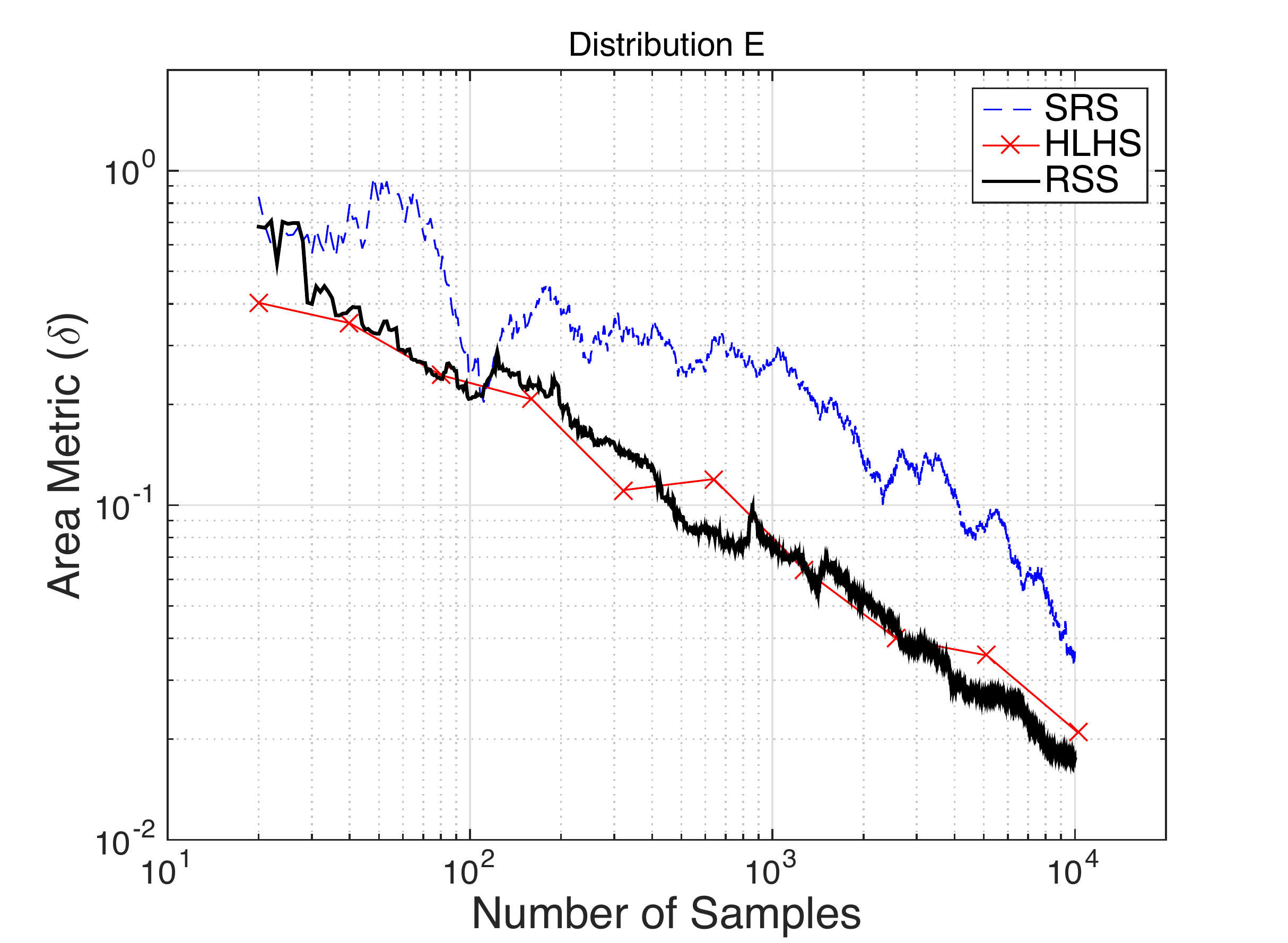}
\includegraphics[width=0.49\columnwidth]{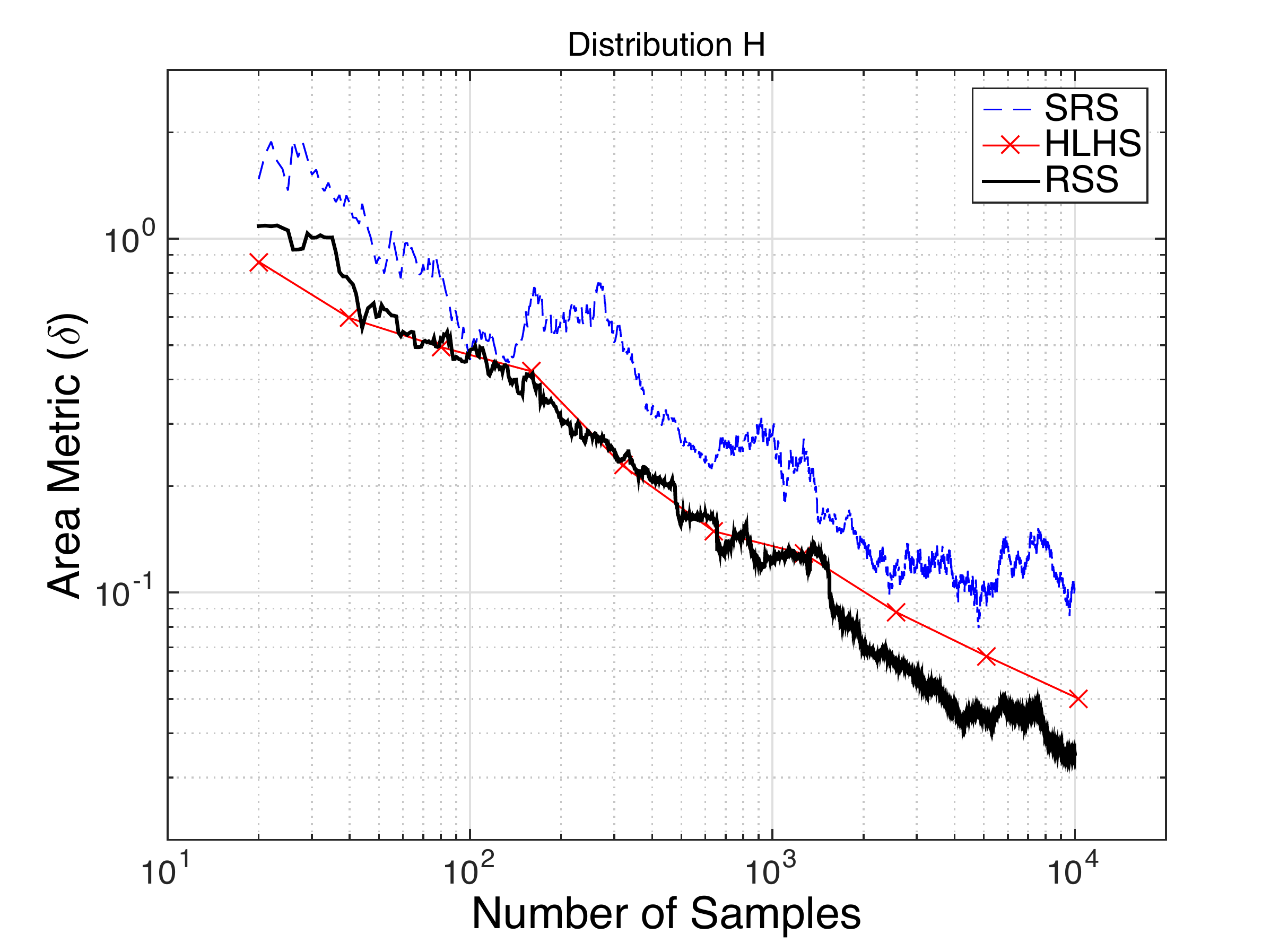}
\includegraphics[width=0.49\columnwidth]{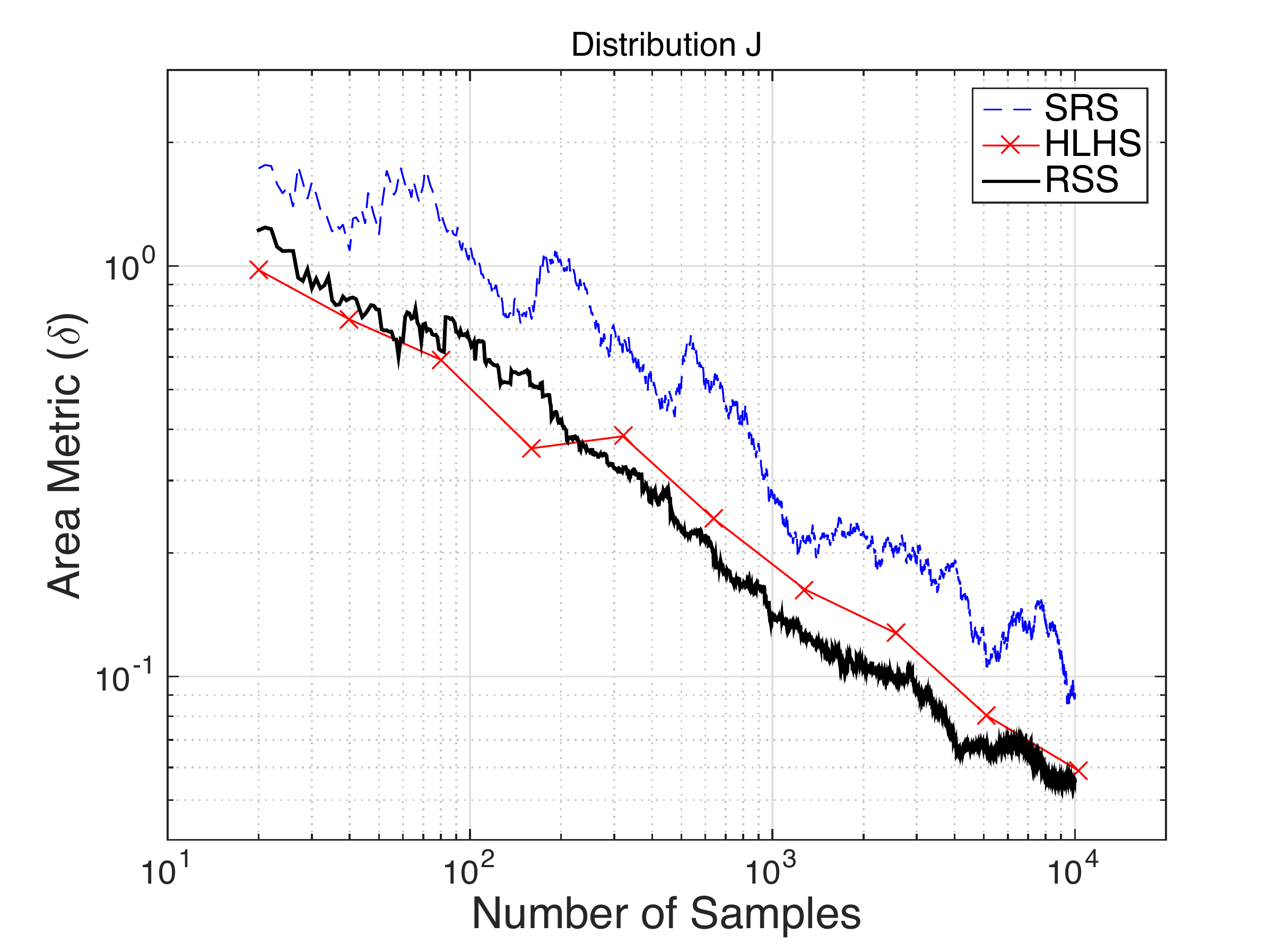}
\caption{Convergence of the empirical CDF to the ``true" CDF for select cases using SRS, HLHS, and RSS.}
\label{fig:5}
\end{figure}


\subsection{Structural response to underwater shock}

Growing emphasis is being placed on uncertainty quantification in computational physical modeling and simulation where it is desirable to evaluate the response variability for large and complex physical systems.  We consider the case of the multi-physical response of a floating structure to underwater shock. These calculations involve large Lagrangian structural finite element models (with tens of thousands of elements) coupled with extremely large Eulerian fluid models (with tens to hundreds of millions of fluid elements) as depicted in Figure \ref{fig:6}(a). Such calculations often require several days to weeks of computation time on massively parallel computers.  An accompanying paper on the validation of computational models for this problem was recently published by the authors \cite{Teferra_et_al_RESS_14}.

\begin{figure}[H]
\centering
\includegraphics[width=1.\columnwidth]{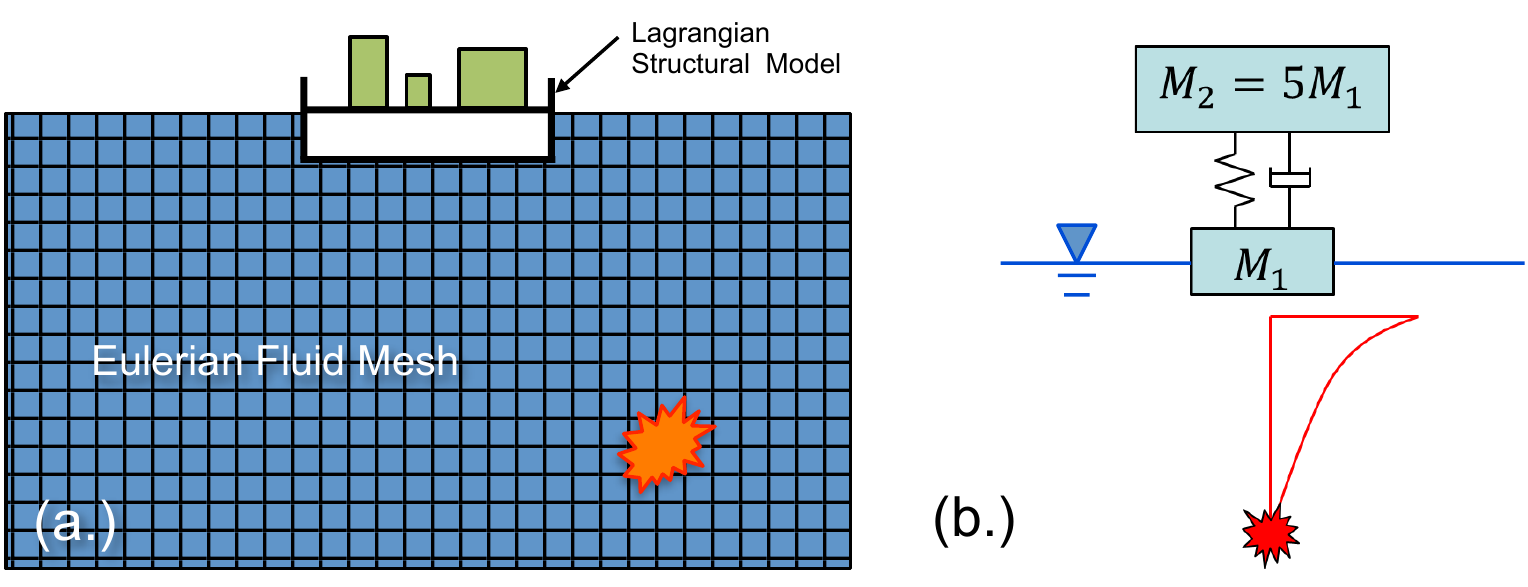}
\caption{(a.) Schematic of coupled fluid structure model for underwater shock response of a surface vessel. (b.) Simplified model used for demonstration purposes. }
\label{fig:6}
\end{figure}

To demonstrate the tractability of uncertainty quantification for these problems, it was necessary to show that statistical convergence could be accomplished for the response in a small number of samples. To do this, a simplified representation of the system was constructed as shown in Figure \ref{fig:6}(b.). The simplified model is an uncoupled one dimensional, two degree of freedom damped system consisting of a mass $M_1$ representing a primary structure supporting a much larger piece of equipment with $M_2=5M_1$. Simplified representations of this form have been used for decades to study fluid-structure interaction in near-surface underwater shock. Early work by Bleich and Sandler \citep{Bleich_Sandler_IJSS_70} presented the effects of bulk surface cavitation on the velocity response of a floating structure subject to underwater shock in a bilinear cavitating fluid. This was followed by the development of a numerically uncoupled treatment of the surface shock response problem in 1D that accounts for the effects of cavitation by DiMaggio et al.\ \cite{DiMaggio_et_al_JAM_81}. In a few more recent papers describing a numerical technique for capturing the effects of cavitation in finite elements models, Sprague and Geers \citep{Sprague_Geers_SaV_01, Sprague_Geers_JCP_03}, with corrections by Stultz and Daddazio \citep{Stultz_Daddazio_SaV_02}, presented the two mass problem (neglecting the effects of structural damping) as it is studied here.  

Unlike the previous authors, we are not specifically concerned with the effects of bulk cavitation and the complex physics that drive the response following the initial kickoff. Instead, we are interested in capturing the variability in peak velocity of the mass $M_2$ given uncertainties in the model using a minimal number of samples. This is motivated by the need to verify that simulation-based uncertainty quantification is a tractable option for the more complex multi-physics model. By considering the sensitivity of the model to various parameters and assessing realistic uncertainties associated with this problem, we identify two uncertain quantities. First, damping of real structures is notoriously difficult to ascertain and model with any accuracy. This is further complicated by the dissipation of energy resulting from fluid-structure interaction. Linear viscous damping ($\xi_d$) is assumed and considered to possess a truncated (non-negative) normal distribution with mean $\mu_d=0.025$ and standard deviation $\sigma_d=0.01$. Second, noticeable variabilities are known to arise in the peak pressure produced from underwater detonation of certain explosives. Here, we consider charge weight to be a normal random variable with mean $\mu_c=117$ lb.\ TNT and standard deviation $\sigma_c=1.17$ lb.\ TNT to account for this variability in peak pressure.


\subsubsection{Modified bootstrap convergence evaluation}
To evaluate statistical convergence, a bootstrap procedure is used that enables us to estimate approximate confidence bounds for response statistics of interest. The traditional bootstrap procedure, developed by Efron \citep{Efron_AoS_79}, performs a random resampling (with replacement) of a dataset to obtain a large number of surrogate sample sets. Statistical analysis of these so-called bootstrap samples provides an estimate of the variability of the computed statistic. The traditional bootstrap method resamples from the original dataset with each sample possessing equal probability of occurrence. That is, the probability of selecting sample $S_l$ from a set of $N$ samples is given by $P[S_l]=1/N$. However, in stratified samples, and more specifically the RSS procedure developed herein, the samples possess associated weights that are not necessarily equal. To account for this unequal weighting (imbalanced stratified design), a modified bootstrap method is proposed as follows:
\begin{enumerate}
\item Find the least common denominator $D$ for all sample weights and create a modified dataset of size $D$ by duplicating samples such that each sample in the modified set has equal weight.
\item Resample a set of $N$ points (bootstrap sample) with replacement from the modified set.
\item Compute the statistic of interest (bootstrap statistic) from the bootstrap sample.
\item Repeat a large number of times to obtain an empirical distribution for the bootstrap statistic.
\end{enumerate}
The proposed resampling process accounts for the sample weights but, it should be emphasized that the bootstrap datasets produced from this method \emph{are not} stratified samples. In any given bootstrap resample, it is likely that certain strata will not be represented. The consequence of this loss of stratification is a potential increase in the variance of the bootstrap statistics. 

This new bootstrap method provides an admittedly imperfect monitor of convergence. The shortcomings of the proposed monitor motivate the need for further research to establish a more effective and general monitor.


\subsubsection{Results}

The bootstrap procedure outlined in the previous section is used to determine confidence intervals for the response statistics from the empirical CDF of the bootstrap statistics. We are interested in evaluating statistics of the peak velocity of mass $M_2$ denoted by $V_2$. Convergence is defined such that the bootstrap 95\% confidence bounds for each statistic of interest are sufficiently narrow. Denoting a computed response statistic by $T$ with bootstrap 95\% confidence intervals given by $T_{95}^l$ and $T_{95}^u$ corresponding to the lower and upper bound respectively, the convergence criterion is expressed as follows:
\begin{linenomath}
\begin{equation}
\label{eqn:11}
\left|\dfrac{T_{95}^u-T_{95}^l}{T}\right|\le \epsilon_{th}
\end{equation}
\end{linenomath}
Note that the convergence criterion based on confidence intervals from an empirical CDF was selected in order to account for skewness of the likely non-Gaussian bootstrap CDF. A convergence criterion based on variance of the bootstrap statistics, for example, would not account for this reality.

Convergence was compared for RSS and SRS. Considering the performance comparison present previously and the need for maximal flexibility in minimizing sample size, HLHS was not considered a good candidate for these analyses. Initially 90,000 analyses were performed using SRS and the convergence evaluated. To achieve superior convergence in all statistics, far fewer calculations were required using RSS. In all, 4,000 RSS analyses were conducted although less than this were necessary to match the convergence or SRS. Plots of the first four moments (with bootstrap confidence intervals) as a function of sample number for peak velocity $V_2$ are provided in Figure \ref{fig:7} using both RSS (left) and SRS (right).
\begin{figure}[H]
\centering
\includegraphics[width=0.4\columnwidth]{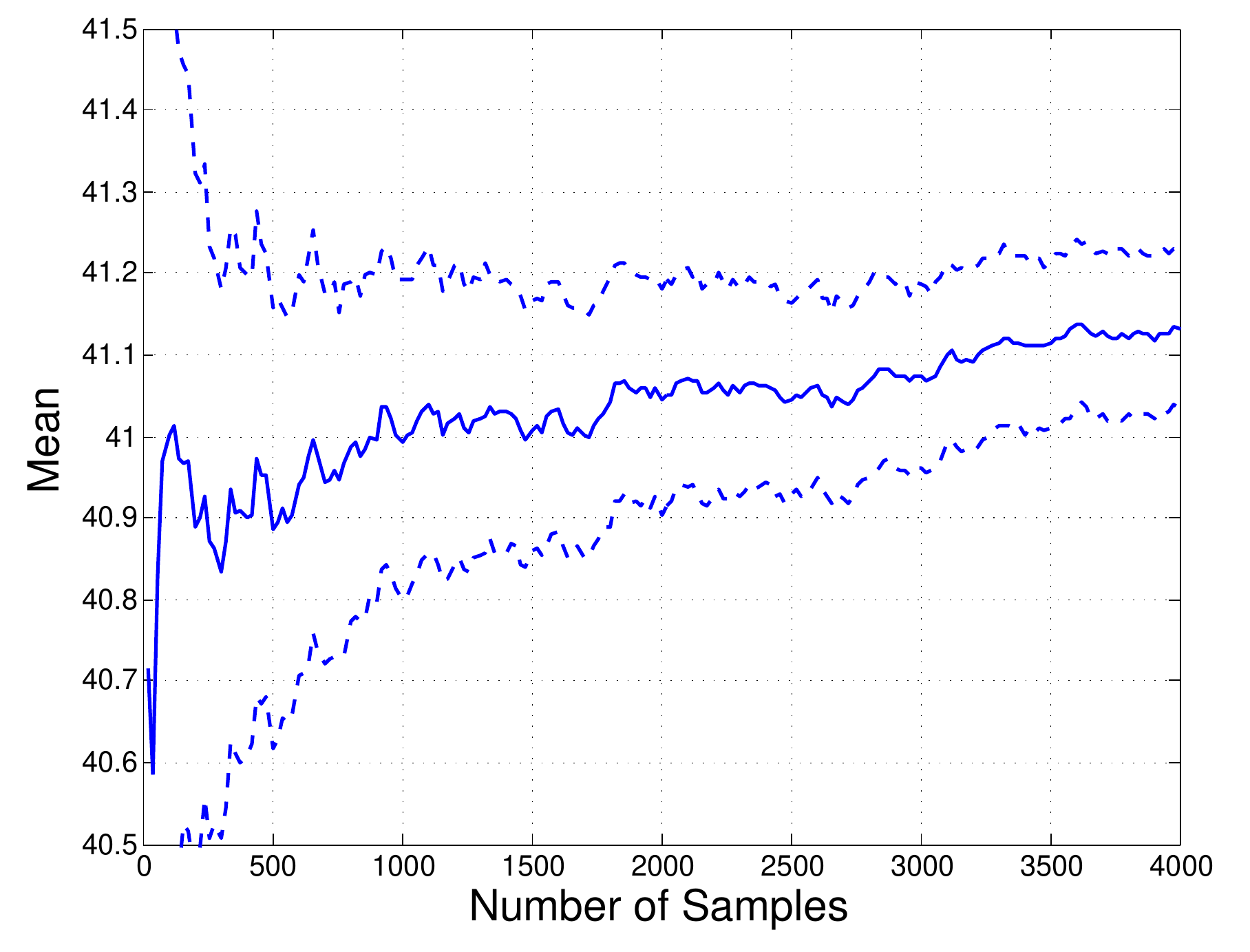}
\includegraphics[width=0.4\columnwidth]{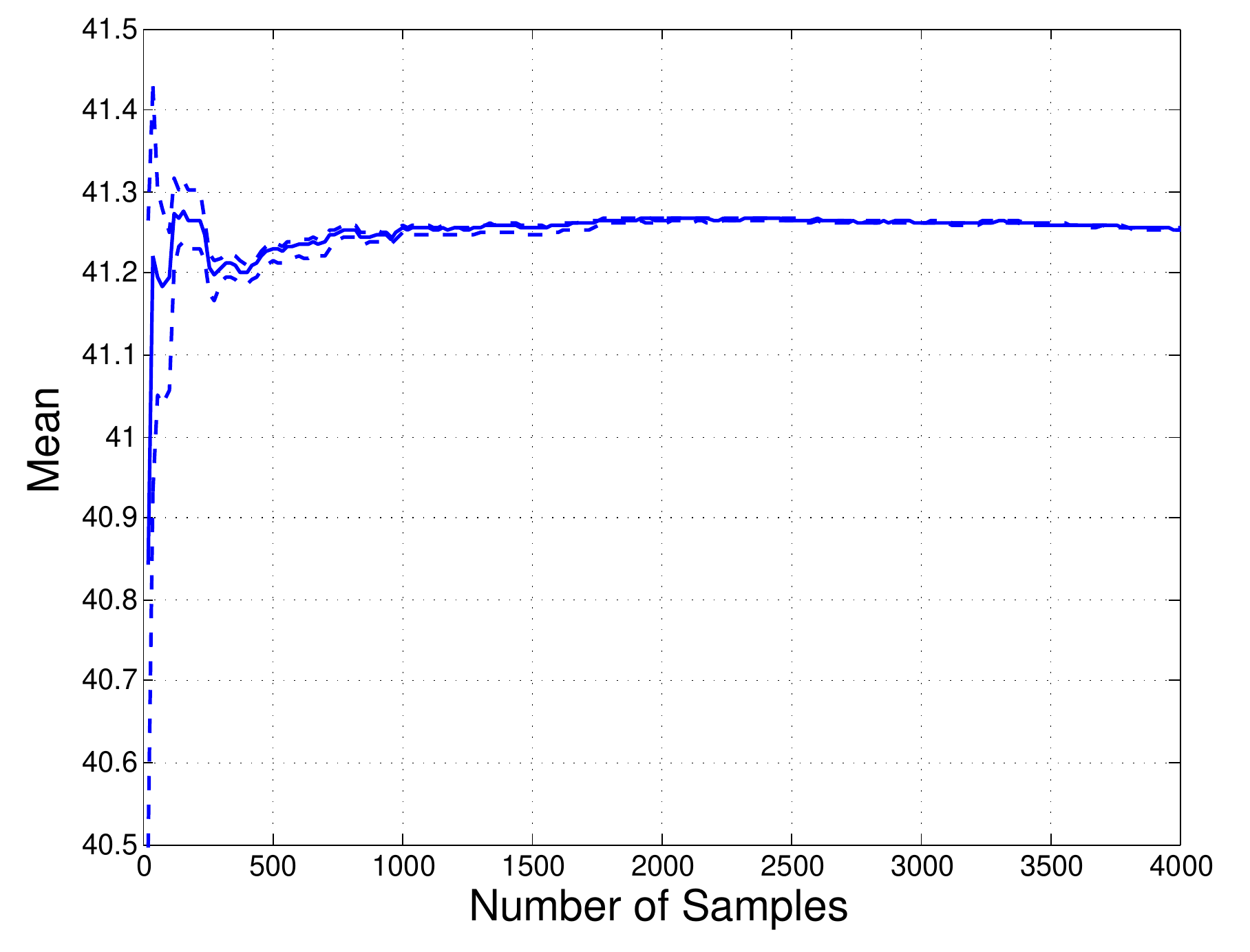}
\includegraphics[width=0.4\columnwidth]{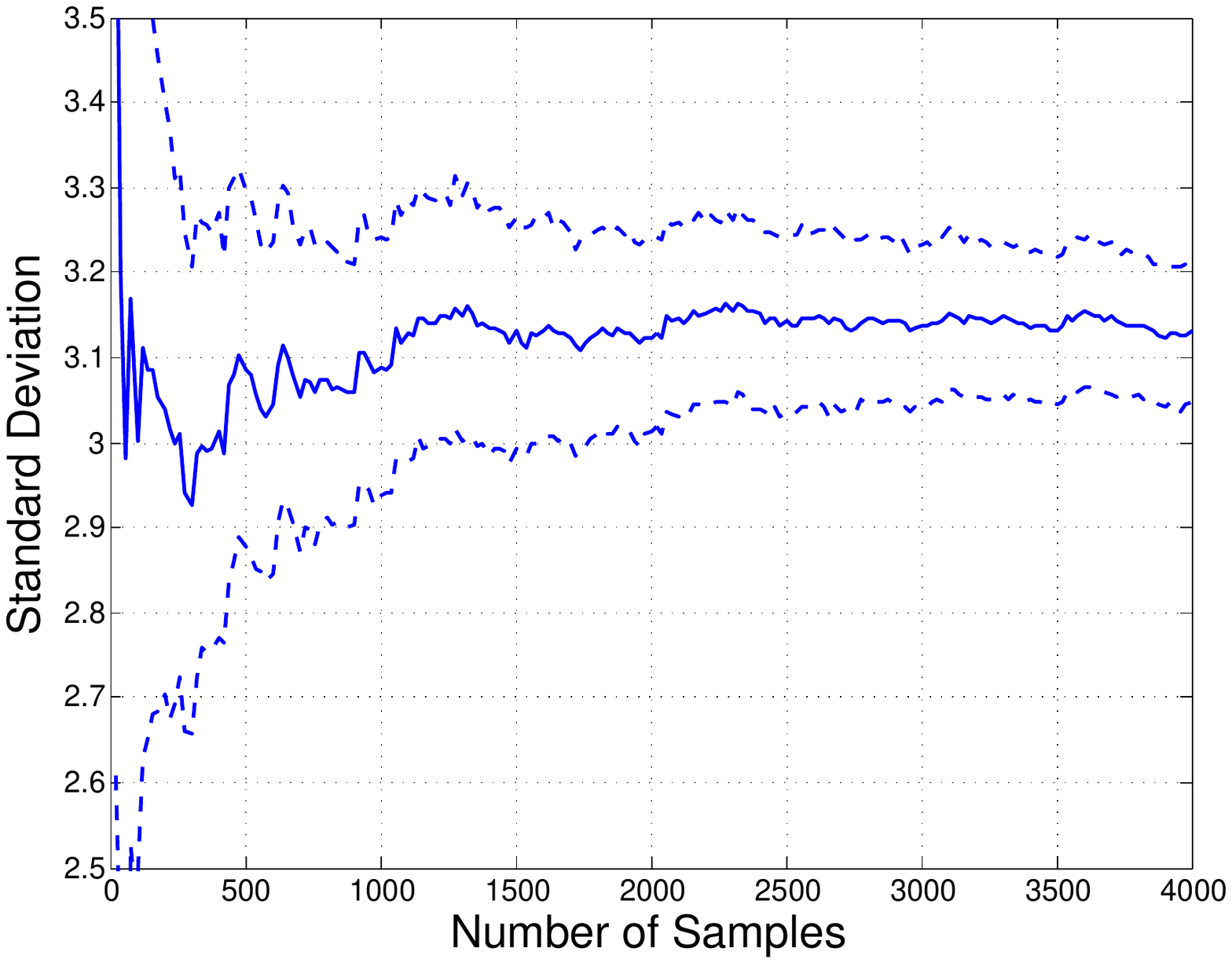}
\includegraphics[width=0.4\columnwidth]{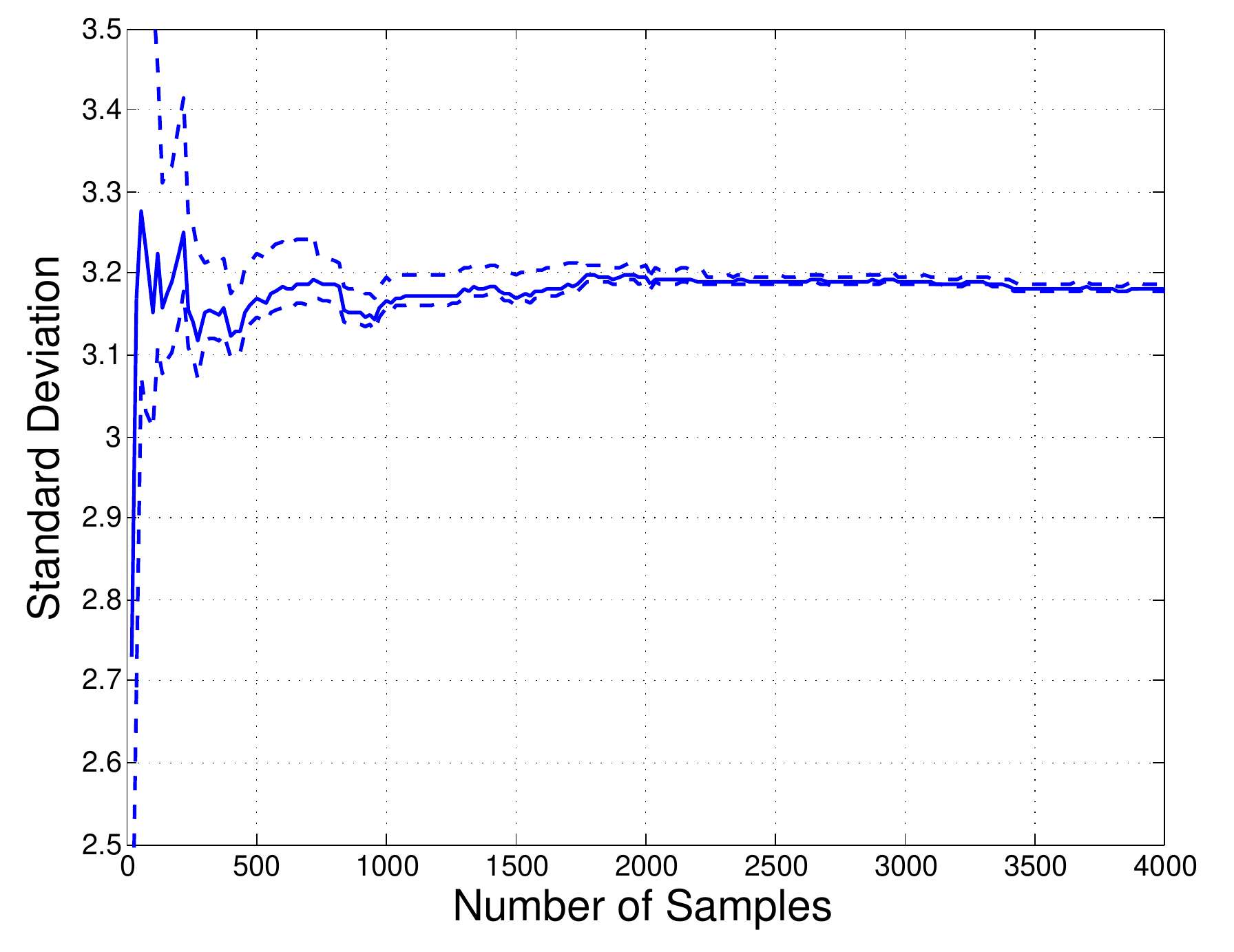}
\includegraphics[width=0.4\columnwidth]{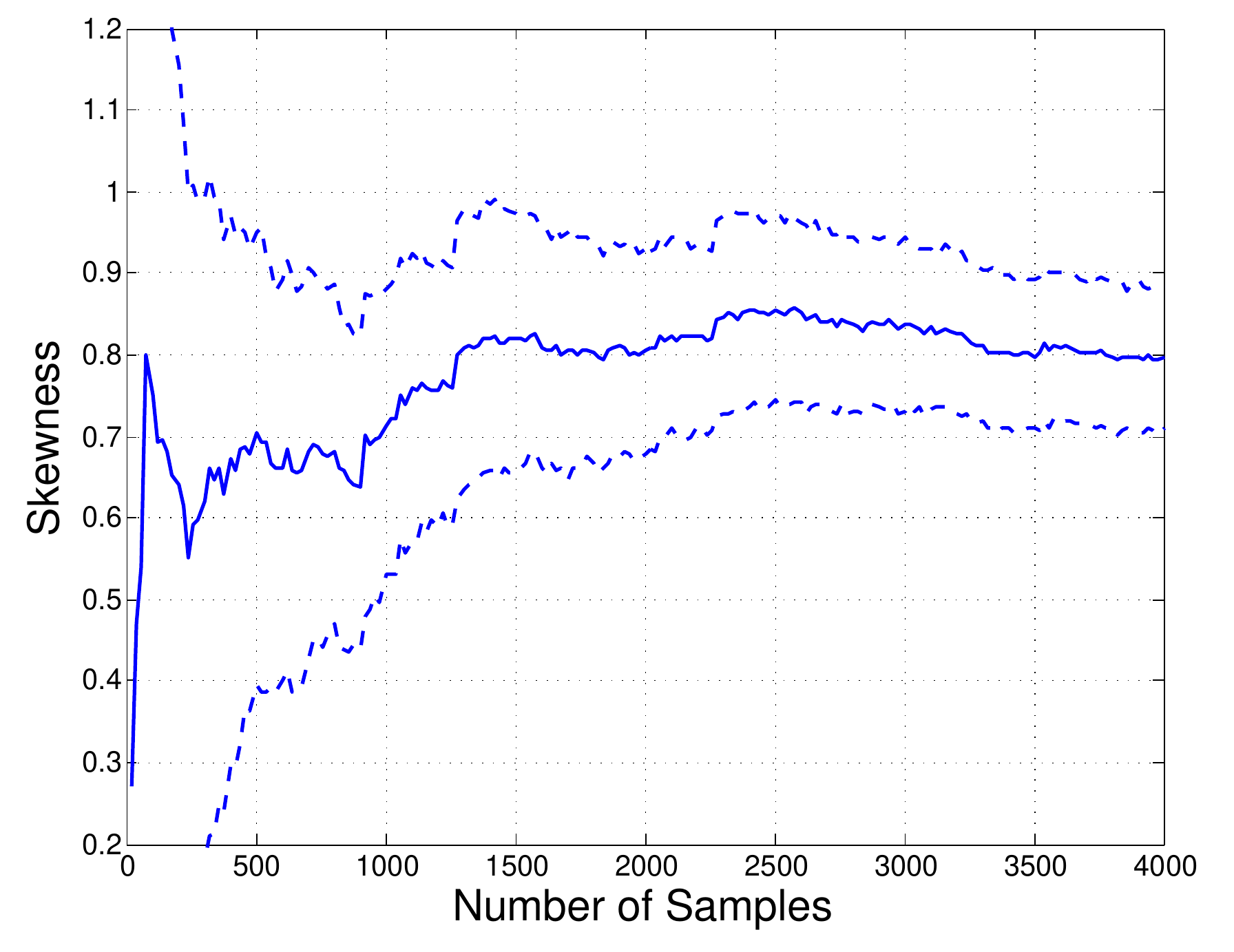}
\includegraphics[width=0.4\columnwidth]{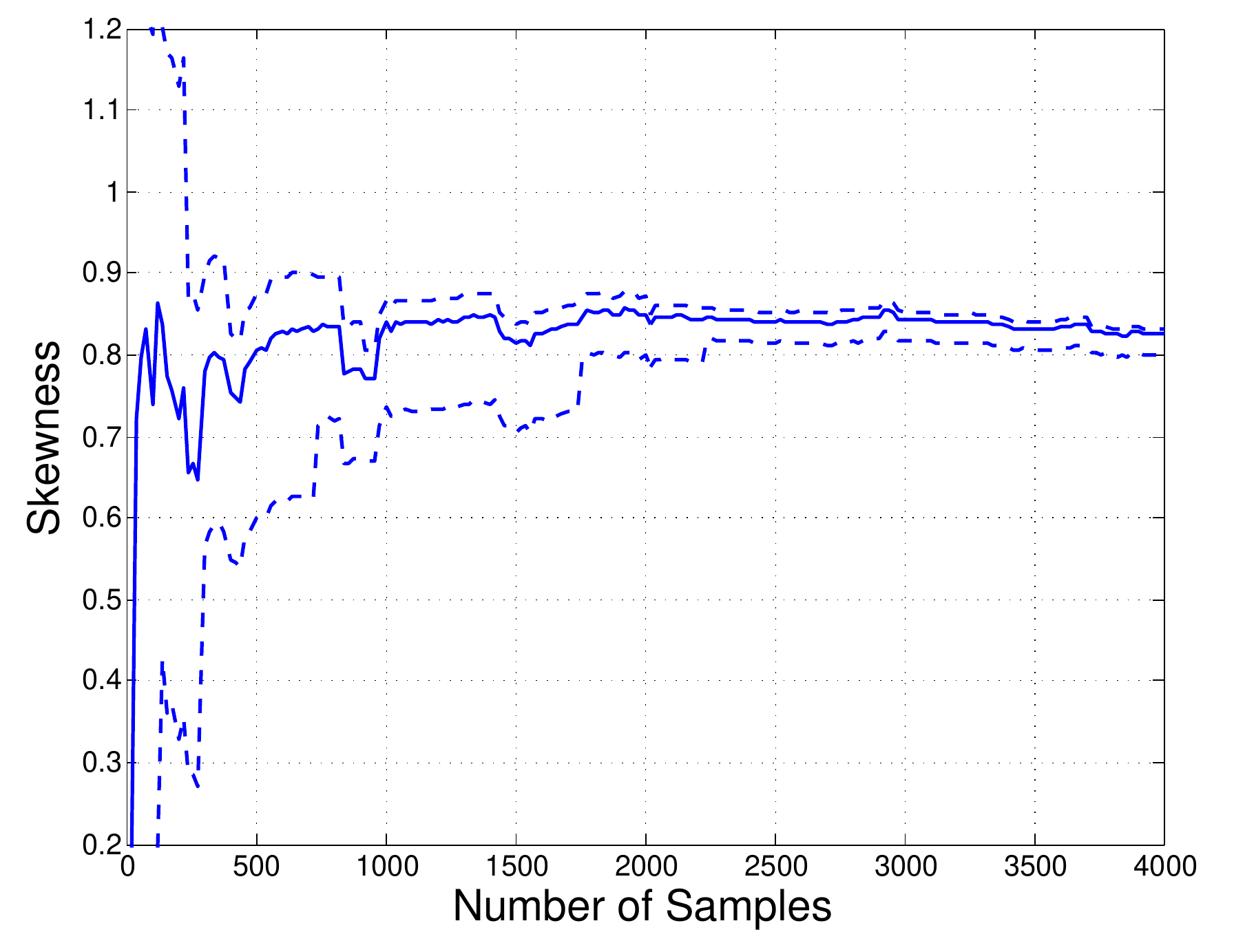}
\includegraphics[width=0.4\columnwidth]{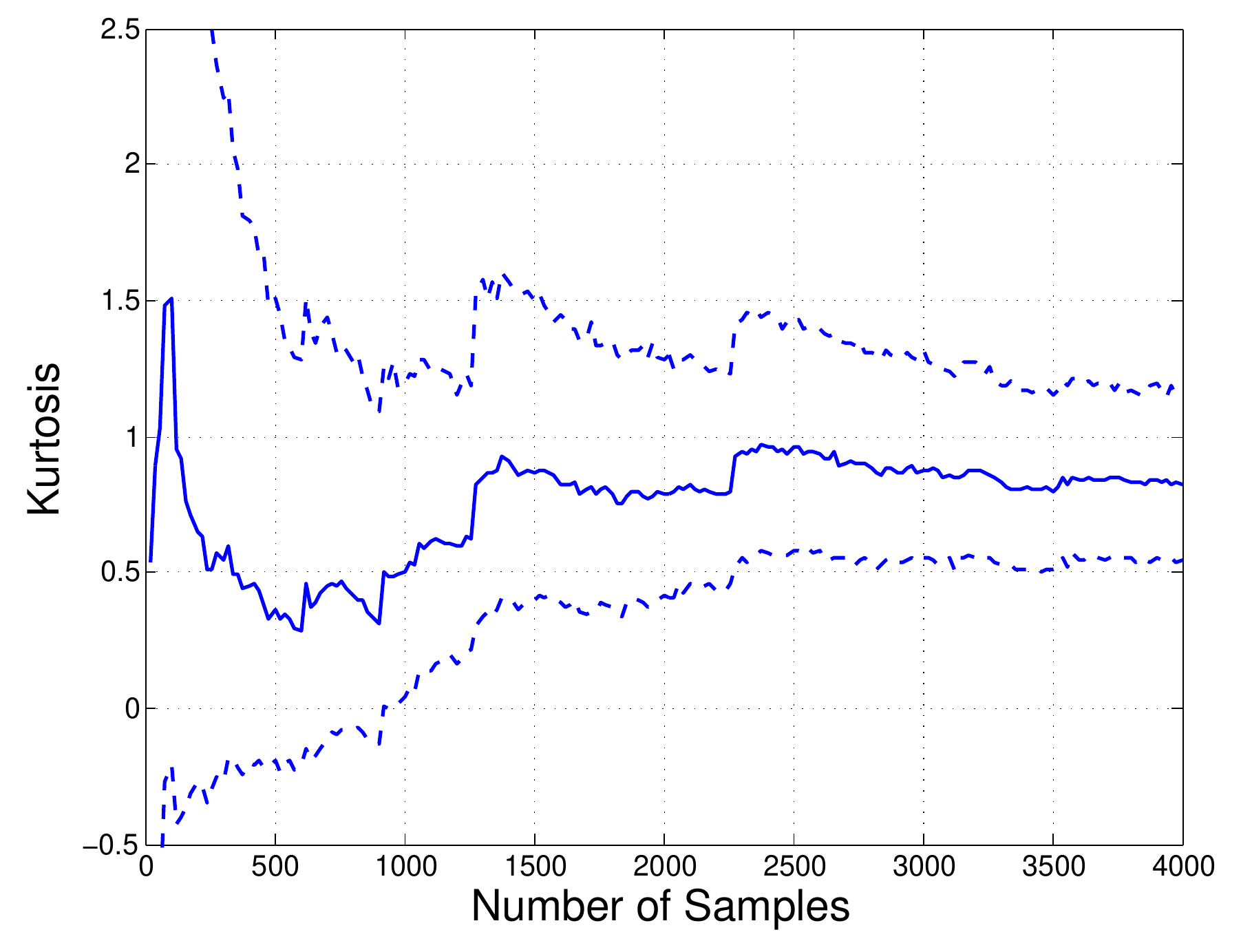}
\includegraphics[width=0.4\columnwidth]{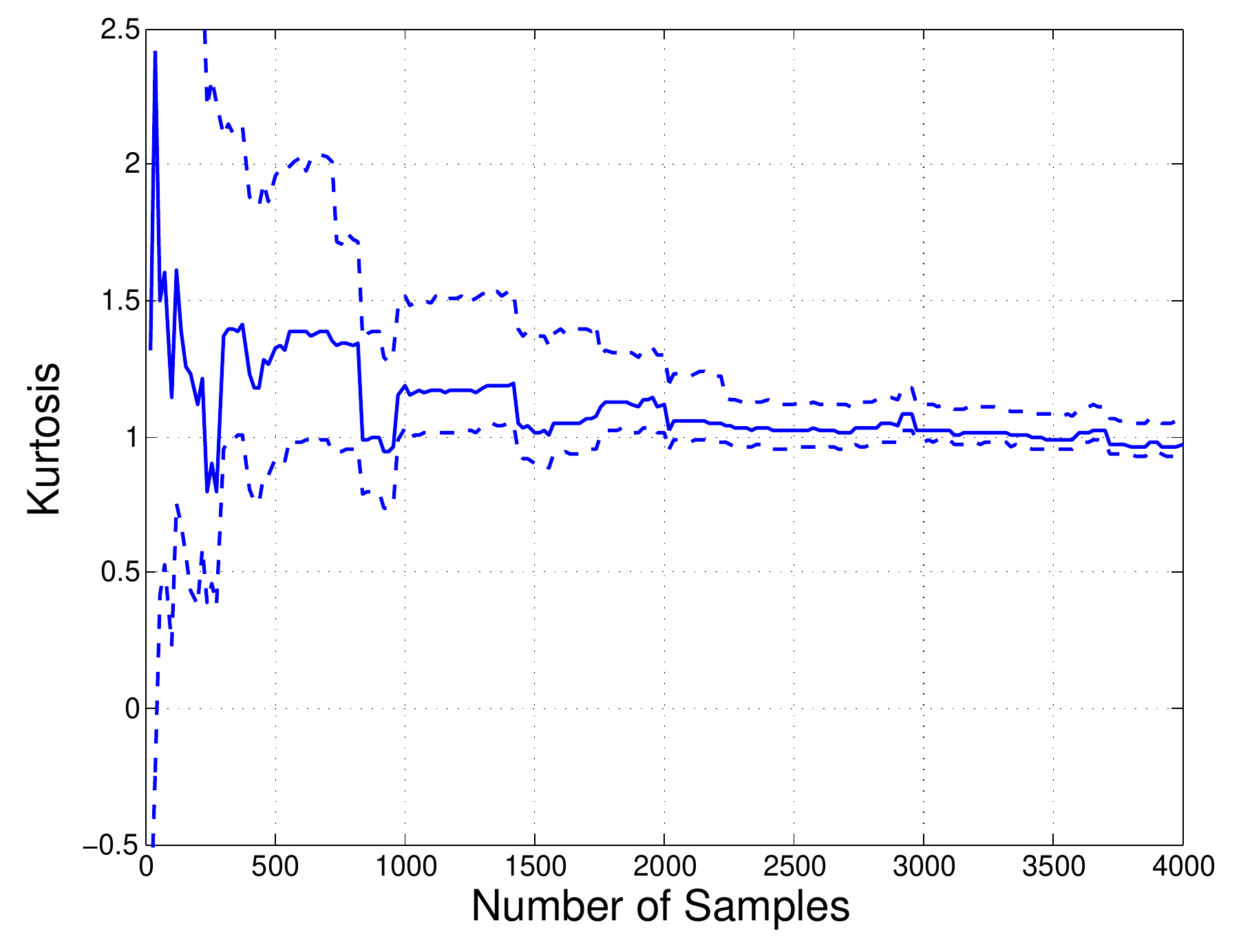}
\caption{Evolution of response statistics with added samples for peak velocity of mass $M_2$ using SRS (left) and RSS (right) showing 95\% bootstrap confidence intervals for 4,000 samples.}
\label{fig:7}
\end{figure}

By establishing convergence thresholds, we determine the number of samples required to produce sample statistics that are considered adequate. Table \ref{tab:3} shows the number of samples required for various convergence criteria of the first four statistical moments for both RSS and SRS. It is noted that the mean value converges extremely rapidly with a cutoff threshold of $0.1\%$ requiring only 300 samples for RSS compared to nearly 85,000 for SRS. Furthermore, a reasonably small error can be gained in standard deviation with only hundreds of samples. Skewness and Kurtosis meanwhile require thousands of samples to converge even using RSS. Nonetheless, the improvement in convergence over SRS is an order of magnitude or better. 
\begin{table}
\caption{Number of samples required to achieve various levels of convergence. -'s indicate that this level of convergence was not obtained in the number of samples performed (4,000 for RSS and 90,000 for SRS).}
\centering
\begin{tabular}{lllllll}
\hline
\bf{$\mathcal{S}_{V_2}$} & \bf{$15\%$} & \bf{$5\%$} & \bf{$2.5\%$} & \bf{$1\%$} & \bf{$0.5\%$} & \bf{$0.1\%$} \\
\hline
$\mu$ & \\
\hspace{3pt} SRS & $<20$ & 40 & 140 & 840 & 3,480 & 84,900\\
\hspace{3pt} RSS & $<20$ & $<20$ & $<20$ & $60$ & $140$ & 300\\
$\sigma$ & \\
\hspace{3pt} SRS & $460$ & $4,500$ & $18,800$ & - & - & -\\
\hspace{3pt} RSS & $120$ & $280$ & $460$ & $1,760$ & $3,760$ & - \\
$\gamma_{1}$ & \\
\hspace{3pt} SRS & $10,750$ & $73,350$ & - & - & - & - \\
\hspace{3pt} RSS & $1,760$ & $2,240$ & - & - & - & - \\
$\gamma_{2}$ & \\
\hspace{3pt} SRS & $69,700$ & - & - & - & - & -\\
\hspace{3pt} RSS & $2,920$ & - & - & - & - & -\\
\hline
\end{tabular}
\label{tab:3}
\end{table}

These data are encouraging when considering the prospects for an uncertainty quantification study for the large-scale multi-physics problem, particularly if the study is interested in evaluating low-order statistics (mean and standard deviation) of various response quantities. As indicated by the plots in Figure \ref{fig:7}, the method converges very rapidly for these statistics even when the response quantity of interest (here peak velocity of mass $M_2$) deviates appreciably from Gaussian. As a result, this large-scale multi-physics study was performed but details of this study are beyond the scope of this paper.

%
%

\section{Conclusions}
In this work, we promote an adaptive approach to Monte Carlo-based UQ rooted in stratified sampling. To motivate the advantages of SS for certain classes of problems the space-filling, orthogonality, and projective properties of SS are studied and compared with simple random sampling (SRS) and Latin hypercube sampling (LHS). It is shown that SS provides superior properties to many low-to-moderate dimensional problems - especially when strong variable interactions occur. To enable the adaptive approach, a new sample size extension methodology - Refined Stratified Sampling (RSS) - is proposed that adds samples sequentially by dividing the existing strata. RSS is proven to reduce variance compared to existing sample size extension methods for SS that add samples to existing strata and is shown to possess comparable or improved convergence when compared to existing extension methods for LHS while affording maximal flexibility. Using RSS, sample can be added one-at-a-time while sample size grows exponentially using hierarchical LHS. Optimality of the stratum refinement is examined and its extension to high dimensional applications discussed. Several future research opportunities along these lines are identified that may facilitate significant breakthroughs in sample-based UQ if achieved. In particular, the potential ability to adaptively identify a partitioning (stratification) of the input space that optimizes the partitioning of the output space could serve to minimize sample size dramatically compared to the current state-of-the-art. Additionally, the RSS method, when combined with new methods for decomposing the sample space into low-dimensional subspaces will allow its extension to very high dimensional problems.

Two examples are presented. The first involved the forward propagation of uncertainty through a low-dimensional stochastic function where convergence can be explicitly evaluated at each sample size extension. The second involved a real multi-physics computational model possessing uncertain parameters for the response of a floating structure subject to an underwater shock. In the second example, a new bootstrap procedure was proposed for resampling from a stratified design to compute confidence intervals for response statistics.

\section*{Appendix}

Example 1 analyzes the number of simulations required to achieve convergence for the 3-dimensional stochastic function given in Eq. \eqref{eqn:cubic_function} with 10 different sets of input distributions using SRS, LHS, and RSS. Table \ref{tab:2} provides the results of this study for all three sample size extension methods and each set of input distributions defined in Table \ref{tab:1}. The columns show the number of samples required to achieve different proportions of converged calculations. For example, 10\% of the 1,000 (or 100) calculation sets converged within 72, 40, and 42 samples for Distribution A using SRS, HLHS, and RSS respectively. In general, the RSS method shows a drastic reduction (often greater than an order of magnitude) in the number of samples needed for convergence when compared to SRS and even affords a significant savings when compared to HLHS. In particular, it is apparent that oversampling is necessary with HLHS given its exponentially increasing sample size.

\begin{table}
\caption{Number of samples required to achieve different proportions of convergence from 1,000 sample sets using SRS, HLHS, and RSS.}
\centering
\begin{tabular}{l|llllll}
\hline
 & \multicolumn{6}{c}{SRS}\\\hline
Dist. & 10\% & 25\% & 50\% & 75\% & 90\% & 95\% \\\hline
A & 72 & 156 & 636 & 2,663 & 6,878 & 12,880\\
B & 74 & 199 & 710 & 2,713 & 7,669 & 14,023\\
C & 73 & 189 & 687 & 2,891 & 8,969 & 15,710\\
D & 81 & 217 & 788 & 3,086 & 9,453 & 17,503\\
E & 93 & 237 & 901 & 3,693 & 11,654 & 20,151\\
F & 101 & 276 & 993 & 4,080 & 15,513 & 29,009\\
G & 125 & 371 & 1,777 & 6,938 & 21,263 & 36,146\\
H & 134 & 496 & 2,301 & 10,919 & 36,507 & 71,982\\
I & 185 & 619 & 4,329 & 24,708 & 84,181 & 165,086\\
J & 199 & 1,004 & 7,136 & 32,866 & 115,060 & 236,090\\
\hline\hline
 & \multicolumn{6}{c}{}\\
 & \multicolumn{6}{c}{HLHS}\\\hline
Dist. & 10\% & 25\% & 50\% & 75\% & 90\% & 95\% \\\hline
A & 40 & 80 & 160 & 640 & 1,280 & 2,560\\
B & 40 & 160 & 320 & 1,280 & 2,560 & 5,120\\
C & 80 & 160 & 320 & 1,280 & 2,560 & 5,120\\
D & 40 & 160 & 640 & 1,280 & 2,560 & 5,120\\
E & 80 & 160 & 640 & 1,280 & 2,560 & 5,120\\
F & 160 & 640 & 1,280 & 2,560 & 10,240 & 10,240\\
G & 320 & 1,280 & 2,560 & 5,120 & 10,240 & 20,480\\
H & 640 & 2,560 & 5,120 & 10,240 & 40,960 & 40,960\\
I & 640 & 2,560 & 5,120 & 20,480 & 40,960 & 81,920\\
J & 1,280 & 5,120 & 10,240 & 40,960 & 81,920 & 163,840\\
\hline\hline
& \multicolumn{6}{c}{}\\
 & \multicolumn{6}{|c}{RSS}\\\hline
Dist. & 10\% & 25\% & 50\% & 75\% & 90\% & 95\% \\\hline
A & 42 & 73 & 144 & 322 & 505 & 942\\
B & 48 & 84 & 177 & 358 & 670 & 933\\
C & 51 & 88 & 196 & 387 & 666 & 954\\
D & 53 & 93 & 205 & 392 & 710 & 1,130\\
E & 59 & 113 & 236 & 459 & 881 & 1,339\\
F & 64 & 130 & 278 & 676 & 1,382 & 1,897\\
G & 92 & 219 & 560 & 1,621 & 3,643 & 5,975\\
H & 132 & 368 & 1,056 & 3,842 & 14,628 & 32,475\\
I & 165 & 463 & 1,861 & 7,219 & 29,408 & 65,372\\
J & 228 & 775 & 3,211 & 13,336 & 57,076 & 125,842\\
\hline
\end{tabular}
\label{tab:2}
\end{table}

\section*{References}

\bibliographystyle{elsarticle-num}
\bibliography{bibliography}{}

\begin{thebibliography}{10}
\expandafter\ifx\csname url\endcsname\relax
  \def\url#1{\texttt{#1}}\fi
\expandafter\ifx\csname urlprefix\endcsname\relax\def\urlprefix{URL }\fi
\expandafter\ifx\csname href\endcsname\relax
  \def\href#1#2{#2} \def\path#1{#1}\fi

\bibitem{Janssen_RESS_13}
H.~Janssen, Monte-carlo based uncertainty analysis: Sampling efficiency and
  sampling convergence, Reliability Engineering and System Safety
  109~(123-132).

\bibitem{Tocher_63}
K.~Tocher, The art of simulation, The English Universities Press, 1963.

\bibitem{Helton_Davis_RESS_03}
J.~Helton, F.~Davis, Latin hypercube sampling and the propagation of
  uncertainty in analyses of complex systems, Reliability Engineering and
  System Safety 81 (2003) 23--69.

\bibitem{McKay_et_al_Tech_79}
M.~McKay, R.~Beckman, W.~Conover, A comparison of three methods of selecting
  values of input variables in the analysis of output from a computer code,
  Technometrics 21~(2) (1979) 239--245.

\bibitem{Olsson_et_al_SS_03}
A.~Olsson, G.~Sandberg, O.~Dahlblom, On latin hypercube sampling for structural
  reliability analysis, Structural Safety 25 (2003) 47--68.

\bibitem{Wang_JMD_03}
G.~Wang, Adaptive response surface method using inherited latin hypercube
  design points, Transactions of the ASME, Journal of Mechanical Design 125
  (2003) 210--220.

\bibitem{Stein_Tech_87}
M.~Stein, Large sample properties of simulations using latin hypercube
  sampling, Technometrics 29~(2) (1987) 143--151.

\bibitem{Owen_JRSSB_92}
A.~Owen, A central limit theorem for latin hypercube sampling, Journal of the
  Royal Statistical Society. Series B 54~(2) (1992) 541--551.

\bibitem{Huntington_Lyrintzis_PEM_98}
D.~Huntington, C.~Lyrintzis, Improvements to and limitations of latin hypercube
  sampling, Probabilistic Engineering Mechanics 13~(4) (1998) 245--253.

\bibitem{Fang_Ma_JOC_01}
K.-T. Fang, C.-X. Ma, Wrap-around l2-discrepancy of random sampling, latin
  hypercube, and uniform designs, Journal of Complexity 17 (2001) 608--624.

\bibitem{Fang_et_al_MoC_02}
K.-T. Fang, C.-X. Ma, P.~Winker, Centered l2-discrepancy of random sampling and
  latin hypercube design, and construction of uniform designs, Mathematics of
  Computation 71 (2002) 275--296.

\bibitem{Glasserman_2003}
P.~Glasserman, Monte Carlo Methods in Financial Engineering, Springer, 2003.

\bibitem{Rubinstein_Kroese_08}
R.~Rubinstein, D.~Kroese, Simulation and the Monte-Carlo Method, John Wiley and
  Sons, New York, NY, 2008.

\bibitem{Kalos_Whitlock_2008}
M.~Kalos, P.~Whitlock, Monte Carlo Methods, Wiley, 2008.

\bibitem{Gilman_Conf_68}
M.~Gilman, A brief survey of stopping rules for monte carlo, in: Proceedings of
  the second conference on applications of simulations, New York, NY, USA,
  1968.

\bibitem{Tong_RESS_06}
C.~Tong, Refinement strategies for stratified sampling methods, Reliability
  Engineering and System Safety 91 (2006) 1257--1265.

\bibitem{Sallaberry_et_al_RESS_08}
C.~Sallaberry, J.~Helton, S.~Hora, Extension of latin hypercube samples with
  correlated variables, Reliability Engineering and System Safety 93 (2008)
  1047--1059.

\bibitem{Vorechovsky_et_al_ICOSSAR_13}
M.~Vorechovsky, D.~Novak, R.~Rusina, Sample size extension in stratified
  sampling: Theory and software implementation, in: G.~Deodatis, B.~Ellingwood,
  D.~Frangopol (Eds.), Safety, Reliability, Risk and Life-Cycle Performance of
  Structures and Infrastructures: Proceedings of the 11th International
  Conference on Structural Safety and Reliability (ICOSSAR 2013), IASSAR, CRC
  Press, New York, 2013.

\bibitem{Iman_Conf_81}
R.~Iman, Statistical methods for including uncertainties associated with the
  geologic isolation of radioactive waste which allow for comparison with
  licensing criteria, in: Proceedings of the symposium on uncertainties
  associated with the regulation of the geologic disposal of high-level
  radioactive waste, Gatlinburg, TN, USA, 1981.

\bibitem{McKay_TR_95}
M.~McKay, Evaluating prediction uncertainty, Tech. Rep. NUREG/CR-6311, Los
  Alamos National Laboratory (1995).

\bibitem{Qian_Bio_09}
P.~Qian, Nested latin hypercube designs, Biometrika 96~(4) (2009) 957--970.

\bibitem{Xiong_EO_09}
F.~Xiong, Y.~Xiong, W.~Chen, S.~Yang, Optimizing latin hypercube design for
  sequential sampling of computer experiments, Engineering Optimization 41
  (2009) 793--810.

\bibitem{Rennen_et_al_SMO_10}
G.~Rennen, B.~Husslage, E.~Van~Dam, D.~Hertog, Nested maximin latin hypercube
  designs, Structural and Multidisciplinary Optimization 41 (2010) 371--395.

\bibitem{Lepage_JCP_78}
G.~Lepage, A new algorithm for adaptive multidimensional integration, Journal
  of Computational Physics 27 (1978) 192--203.

\bibitem{Press_Farrar_CiP_90}
W.~Press, G.~Farrar, Recursive stratified sampling for multidimensional monte
  carlo integration, Computers in Physics 4 (1990) 190--195.

\bibitem{Shields_Sundar_IJRS_2015}
M.~Shields, V.~Sundar, Targeted random sampling: A new approach for efficient
  reliability estimation of complex systems, International Journal of
  Reliability and Safety.

\bibitem{Johnson_et_al_JSPI_90}
M.~Johnson, L.~Moor, D.~Ylvisaker, Minimax and maximin distance designs,
  Journal of Statistical Planning and Interence 26 (1990) 131--148.

\bibitem{Morris_Mitchell_JSPI_95}
M.~Morris, T.~Mitchell, Exploratory designs for computational experiments,
  Journal of Statistical Planning and Interence 43 (1995) 381--402.

\bibitem{Ye_et_al_JSPI_2000}
K.~Ye, W.~Li, A.~Sudjianto, Algorithmic construction of optimal symmetric latin
  hypercube designs, Journal of Statistical Planning and Interence 90 (2000)
  145--159.

\bibitem{Joseph_Hung_SS_08}
V.~Joseph, Y.~Hung, Orthogonal-maximin latin hypercube designs, Statistica
  Sinica 18~(171-186).

\bibitem{Liefvendahl_Stocki_JSPI_06}
M.~Liefvendahl, R.~Stocki, A study on algorithms for optization of latin
  hypercubes, Journal of Statistical Planning and Interence 136~(9) (2006)
  3231--3247.

\bibitem{Park_JSPI_94}
J.~Park, Optimal latin-hypercube designs for computer experiments, Journal of
  Statistical Planning and Interence 39 (1994) 95--111.

\bibitem{Iman_Conover_CSSC_82}
R.~Iman, W.~Conover, A distribution-free approach to inducing rank correlation
  among input variables, Communications in Statistics: Simulation and
  Computation 11~(3) (1982) 311--334.

\bibitem{Florian_PEM_92}
A.~Florian, An efficient sampling scheme: Updated latin hypercube sampling,
  Probabilistic Engineering Mechanics 7 (1992) 123--130.

\bibitem{Vorechovsky_Novak_PEM_09}
M.~Vorechovsky, D.~Novak, Correlation control in small-sample monte carlo type
  simulations i: A simulated annealing approach, Probabilistic Engineering
  Mechanics 24 (2009) 452--462.

\bibitem{Tang_JASA_93}
B.~Tang, Orthogonal array-based latin hypercubes, Journal of the American
  Statistical Association 88 (1993) 1392--1397.

\bibitem{Ye_JASA_98}
K.~Ye, Orthogonal column latin hypercubes and their application in computer
  experiments, Journal of the American Statistical Association 93 (1998)
  1430--1439.

\bibitem{Cioppa_Lucas_Tech_07}
T.~Cioppa, T.~Lucas, Efficient nearly orthogonal and space-filling latin
  hypercubes, Technometrics 49~(1) (2007) 45--55.

\bibitem{Fang_et_al_06}
K.-T. Fang, R.~Li, A.~Sudjianto, Design and Modeling for Computer Experiments,
  Chapman and Hall/CRC, London, UK, 2006.

\bibitem{Dalbey_Karystinos_AIAA_10}
K.~Dalbey, G.~Karystinos, Fast generation of space-filling latin hypercube
  sample designs, in: Proceedings of the 13th AIAA/ISSMO Multidisciplinary
  Analysis and Optimization Conference, 2010.

\bibitem{Hickernell_98}
F.~Hickernell, Random and Quasi-Random Point Sets, Springer, 1998, Ch. Lattice
  rules: how well do they measure up?

\bibitem{Aurenhammer_ACS_91}
F.~Aurenhammer, Voronoi diagrams - a survey of a fundamental geometric data
  structure, ACM Computing Surveys 23~(3) (1991) 345--405.

\bibitem{Wu_Hamada_00}
C.~Wu, M.~Hamada, Experiments: Planning, Analysis, and Parameter Design
  Optimization, Wiley, New York, 2000.

\bibitem{Saltelli_et_al_08}
A.~Saltelli, M.~Ratto, T.~Andres, F.~Campolongo, J.~Cariboni, D.~Gatelli,
  M.~Saisana, S.~Tarantola, Global Sensitivity Analysis: The Primer, Wiley,
  2008.

\bibitem{Grigoriu_AMM_09}
M.~Grigoriu, Reduced order models for random function. application to
  stochastic problems, Applied Mathematical Modeling 33 (2009) 161--175.

\bibitem{Roy_Oberkampf_CMAME_11}
C.~Roy, W.~Oberkampf, A comprehensive framework for verification, validation,
  and uncertainty quantification in scientific computing, Computer Methods in
  Applied Mechanics and Engineering 200 (2011) 2131--2144.

\bibitem{Teferra_et_al_RESS_14}
K.~Teferra, M.~Shields, A.~Hapij, R.~Daddazio, Mapping model validation metrics
  to subject matter expert scores for model adequacy assessment, Reliability
  Engineering and System Safety 132 (2014) 9--19.

\bibitem{Bleich_Sandler_IJSS_70}
H.~Bleich, I.~Sandler, Interaction between structures and bilinear fluids,
  International Journal of Solids and Structures 6 (1970) 617--639.

\bibitem{DiMaggio_et_al_JAM_81}
F.~DiMaggio, I.~Sandler, D.~Rubin, Uncoupling approximations in fluid-structure
  interaction problems with cavitation, Journal of Applied Mechanics 48 (1981)
  753--756.

\bibitem{Sprague_Geers_SaV_01}
M.~Sprague, T.~Geers, Computational treatments of cavitation effects in
  near-free-surface underwater shock, Shock and Vibration 7 (2001) 105--122.

\bibitem{Sprague_Geers_JCP_03}
M.~Sprague, T.~Geers, Spectral elements and field separation for an acoustic
  field subject to cavitation, Journal of Computational Physics 184 (2003)
  149--162.

\bibitem{Stultz_Daddazio_SaV_02}
K.~Stultz, R.~Daddazio, The applicability of fluid-structure interaction
  approaches to the analysis of floating targets subjected to undex loading,
  in: Proceedings of the 73rd Shock and Vibration Symposium, Newport, RI, 2002.

\bibitem{Efron_AoS_79}
B.~Efron, Bootstrap methods: Another look at the jackknife, The Annals of
  Statistics 7 (1979) 1--26.

\end{thebibliography}

\end{document}